\documentclass[preprint]{elsarticle}

\usepackage{lineno,hyperref}
\usepackage{graphicx,amsmath,amssymb,amsthm,bm}
\modulolinenumbers[5]

\journal{}

\newtheorem{definition}{Definition}
\newtheorem{proposition}{Proposition}
\newtheorem{lemma}{Lemma}
\newtheorem{theorem}{Theorem}

\begin{document}

\begin{frontmatter}

\title{On the existence of fair zero-determinant strategies in the periodic prisoner's dilemma game}

\author{Ken Nakamura}
\address{Graduate School of Sciences and Technology for Innovation, Yamaguchi University, Yamaguchi 753-8512, Japan}

\author{Masahiko Ueda\corref{mycorrespondingauthor}}
\ead{ueda@tmu.ac.jp}
\address{Department of Physics, Tokyo Metropolitan University, Tokyo 192-0397, Japan}

\begin{abstract}
Stochastic games are a framework for investigating long-term interdependence of multi-agent systems with environmental feedback.
When the number of environmental states is one, they are reduced to repeated games.
In repeated games, zero-determinant (ZD) strategies attract much attention in evolutionary game theory, since they can unilaterally control payoffs.
Especially, fair ZD strategies unilaterally equalize the payoff of the focal player and the average payoff of the opponents, and they were found in several games including the social dilemma games.
Although the existence condition of ZD strategies in repeated games was specified, its extension to stochastic games remains largely unclear.
Here, we investigate the existence condition of fair ZD strategies in the periodic prisoner's dilemma game, which is one of the simplest stochastic games.
The periodic prisoner's dilemma game consists of two environmental states and the two states alternate deterministically.
Whereas each stage game is not necessarily the prisoner's dilemma game, the whole game can be regarded as the prisoner's dilemma game on average.
We show that fair ZD strategies do not necessarily exist in the periodic prisoner's dilemma game, in contrast to the repeated prisoner's dilemma game.
Furthermore, we also prove that the Tit-for-Tat strategy, which imitates the opponent's action, is not necessarily a fair ZD strategy in the periodic prisoner's dilemma game, whereas the Tit-for-Tat strategy is always a fair ZD strategy in the repeated prisoner's dilemma game.
Our results highlight difference between ZD strategies in the periodic prisoner's dilemma game and those in the standard repeated prisoner's dilemma game.
\end{abstract}

\begin{keyword}
Repeated games; Zero-determinant strategies; Stochastic games; Payoff control; Prisoner's dilemma
\end{keyword}

\end{frontmatter}


\section{Introduction}
\label{sec:introduction}
Repeated games are a framework for investigating long-term interdependence of multi-agent systems \cite{MaiSam2006}.
Agents can adopt strategies according to all previous actions of all agents.
In repeated games, zero-determinant (ZD) strategies attract much attention in evolutionary game theory \cite{PreDys2012}.
ZD strategies unilaterally enforce linear relationships between payoffs, and they can be used to control multi-agent systems \cite{HCN2018}.
Especially, fair ZD strategies unilaterally equalize the payoff of the focal player and the average payoff of the opponents \cite{PreDys2012,HWTN2014}.
In the repeated prisoner's dilemma game, the Tit-for-Tat (TFT) strategy \cite{RCO1965,AxeHam1981}, which imitates the opponent's previous action, is a fair ZD strategy which unilaterally equalizes the payoffs of two players \cite{PreDys2012}.
In two-player games, because a fair ZD strategy can invade any other strategies by neutral drift in evolutionary game theory, it can be successful in evolution.
Furthermore, in two-player games, if one player adopts a fair ZD strategy, it incentivizes the opponent to optimize the payoffs of both players \cite{Aki2016}.
So far, the existence of fair ZD strategies has been proved in the prisoner's dilemma game \cite{PreDys2012}, the public goods game \cite{HWTN2014,PHRT2015}, continuous donation game \cite{McAHau2016}, two-player potential games \cite{Ued2022}, two-player games without generalized rock-paper-scissors cycles \cite{Ued2022b}, and the Cournot oligopoly game \cite{UYI2025}.
The existence condition of general ZD strategies in repeated games was completely specified \cite{Ued2022b}.

Stochastic games are an extension of repeated games, where a state of an environment exists, and the state changes to another one according to an action profile of players \cite{Sha1953}.
Recently, evolution of cooperation in stochastic games again attracts attention \cite{HSCN2018}.
If transition to a worse state is coupled to defection in the prisoner's dilemma, mutual cooperation can be achieved more easily than in the standard repeated prisoner's dilemma game.
Furthermore, performance of ZD strategies in stochastic games has gradually been investigated \cite{DRWW2021,LiuWu2022,MMHet2025}.
However, since stochastic games are more complicated than repeated games, the existence condition of ZD strategies has not been specified yet.

In Ref. \cite{MMHet2025}, McAvoy and coworkers provided one of the simplest stochastic games.
In this stochastic game, two environmental states exist, and two stage games are alternately played.
Whereas each stage game is not necessarily the prisoner's dilemma game, the whole game can be regarded as the prisoner's dilemma game on average.
Here we call this stochastic game as the \emph{periodic prisoner's dilemma game}.
Although the periodic prisoner's dilemma game is simple, the existence condition of ZD strategies has not been specified.
In particular, while the existence of fair ZD strategies was numerically found in Ref. \cite{MMHet2025}, properties of such fair ZD strategies remain largely unclear.

In this paper, we investigate the existence condition of fair ZD strategies in the periodic prisoner's dilemma game.
Especially, we provide a necessary and sufficient condition for the existence of fair ZD strategies.
Furthermore, we also specify the condition where TFT becomes a fair ZD strategy in the game.
These results highlight difference between the periodic prisoner's dilemma game and the standard repeated prisoner's dilemma game.

The paper is organized as follows.
In Section \ref{sec:model}, we introduce the periodic prisoner's dilemma game.
In Section \ref{sec:preliminaries}, we explain properties of ZD strategies in general stochastic games.
In Section \ref{sec:results}, we provide our main results on the existence of fair ZD strategies in the periodic prisoner's dilemma game.
Section \ref{sec:conclusion} is devoted to concluding remarks.

\section{Model}
\label{sec:model}
We introduce a stochastic game $G:=\left( \mathcal{N}, \Sigma, \left\{ A_j \right\}_{j\in \mathcal{N}}, T_\mathrm{E}, P_\mathrm{E}^{(1)}, \left\{ s_j \right\}_{j\in \mathcal{N}} \right)$ \cite{Sha1953,HSCN2018}.
$\mathcal{N}$ is the set of players.
$\Sigma$ is the set of states of an environment.
$A_j$ is the set of actions of player $j$.
$T_\mathrm{E}: \prod_{k\in \mathcal{N}}A_k \times \Sigma \rightarrow \Delta(\Sigma)$ is the transition function of states, where $\Delta(\Sigma)$ is the probability simplex on $\Sigma$.
$P_\mathrm{E}^{(1)}$ is the probability distribution of the initial state.
$s_j: \prod_{k\in \mathcal{N}}A_k \times \Sigma \rightarrow \mathbb{R}$ is the one-shot payoff function of player $j$.
We write $\mathcal{A}:= \prod_{k\in \mathcal{N}} A_k$ and $A_{-j}:= \prod_{k\neq j} A_k$ for all $j \in \mathcal{N}$.
Furthermore, we introduce the notations $\bm{a}:=\left( a_k \right)_{k\in \mathcal{N}} \in \mathcal{A}$ and $a_{-j}:=\left( a_k \right)_{k\neq j} \in A_{-j}$.
The players can choose actions in each round referring to all histories of actions and states, and we call such plans \emph{strategies}.
When we write an action profile and a state at the $t$-th round as $\bm{a}^{(t)}$ and $\sigma^{(t)}$ with $t\geq 1$, respectively, the payoff of player $j$ in the stochastic game is defined by
\begin{align}
 \mathcal{S}_j &:= \lim_{T\rightarrow \infty} \frac{1}{T} \sum_{t=1}^T \mathbb{E} \left[ s_j\left( \bm{a}^{(t)}, \sigma^{(t)} \right) \right],
\end{align}
where $\mathbb{E}[\cdot]$ is the expectation with respect to strategies of all players and the transition function $T_\mathrm{E}$.
In this paper, for a given set $H$, we define the Kronecker delta
\begin{align}
 \delta_{h, h^\prime} &:= \left\{
  \begin{array}{ll}
    1 & \left( h=h^\prime \right) \\
    0 & \left( h\neq h^\prime \right)
  \end{array}
  \right.
\end{align}
for $h, h^\prime \in H$.
That is, we use the same notation of the Kronecker delta for the case $H=\Sigma$ and for the cases $H=A_j$ and so on.

As a special example of stochastic games, we introduce a periodic prisoner's dilemma game \cite{MMHet2025}.
This model is a two-state stochastic game in which the two stage games alternate deterministically.
The sets are defined as $\mathcal{N}= \{ 1, 2 \}$, $\Sigma= \{ \sigma_1, \sigma_2 \}$, $A_j= \left\{ C, D \right\}$ $(j=1, 2)$.
The transition function is defined as
\begin{align}
 T_\mathrm{E} \left( \sigma | \bm{a}, \sigma_1 \right) &= \delta_{\sigma, \sigma_2} \quad \left( \forall \bm{a} \in \mathcal{A} \right) \nonumber \\
 T_\mathrm{E} \left( \sigma | \bm{a}, \sigma_2 \right) &= \delta_{\sigma, \sigma_1} \quad \left( \forall \bm{a} \in \mathcal{A} \right).
\end{align}
The probability distribution of the initial state is defined as
\begin{align}
 P_\mathrm{E}^{(1)}\left( \sigma \right) &= \frac{1}{2} \delta_{\sigma, \sigma_1} + \frac{1}{2} \delta_{\sigma, \sigma_2}.
\end{align}
That is, the initial state is chosen randomly, and then two stage games are alternately played.
The one-shot payoffs are defined as in Tables \ref{table:PD1} and \ref{table:PD2}.
\begin{table}[tb]
  \centering
  \caption{Payoffs in state $\sigma_1$.}
  \begin{tabular}{|c|cc|} \hline
    & $C$ & $D$ \\ \hline
   $C$ & $R^{(1)}, R^{(2)}$ & $S^{(1)}, T^{(2)}$ \\
   $D$ & $T^{(1)}, S^{(2)}$ & $P^{(1)}, P^{(2)}$ \\ \hline
  \end{tabular}
  \label{table:PD1}
  
  \centering
  \caption{Payoffs in state $\sigma_2$.}
  \begin{tabular}{|c|cc|} \hline
    & $C$ & $D$ \\ \hline
   $C$ & $R^{(2)}, R^{(1)}$ & $S^{(2)}, T^{(1)}$ \\
   $D$ & $T^{(2)}, S^{(1)}$ & $P^{(2)}, P^{(1)}$ \\ \hline
  \end{tabular}
  \label{table:PD2}
\end{table}
We assume that $T^{(1)}+T^{(2)}>R^{(1)}+R^{(2)}>P^{(1)}+P^{(2)}>S^{(1)}+S^{(2)}$ and $2R^{(1)}+2R^{(2)}>T^{(1)}+T^{(2)}+S^{(1)}+S^{(2)}$ as in the standard prisoner's dilemma game \cite{PreDys2012}.
It should be noted that each game does not need to be the prisoner's dilemma game. 
When we introduce the notation $\bar{R} := \left( R^{(1)} + R^{(2)} \right)/2$, $\delta_R := \left( R^{(1)} - R^{(2)} \right)/2$, and so on, the assumption is rewritten as $\bar{T}>\bar{R}>\bar{P}>\bar{S}$ and $2\bar{R}>\bar{T}+\bar{S}$.
Therefore, there is no assumption on $\left( \delta_R, \delta_S, \delta_T, \delta_P \right)$.
By using this notation, the payoffs are rewritten as
\begin{align}
 s_1\left( C, C, \sigma_1 \right) &= \bar{R} + \delta_R \nonumber \\
 s_1\left( C, D, \sigma_1 \right) &= \bar{S} + \delta_S \nonumber \\
 s_1\left( D, C, \sigma_1 \right) &= \bar{T} + \delta_T \nonumber \\
 s_1\left( D, D, \sigma_1 \right) &= \bar{P} + \delta_P \nonumber \\
 s_1\left( C, C, \sigma_2 \right) &= \bar{R} - \delta_R \nonumber \\
 s_1\left( C, D, \sigma_2 \right) &= \bar{S} - \delta_S \nonumber \\
 s_1\left( D, C, \sigma_2 \right) &= \bar{T} - \delta_T \nonumber \\
 s_1\left( D, D, \sigma_2 \right) &= \bar{P} - \delta_P
\end{align}
and
\begin{align}
 s_2\left( C, C, \sigma_1 \right) &= \bar{R} - \delta_R \nonumber \\
 s_2\left( C, D, \sigma_1 \right) &= \bar{T} - \delta_T \nonumber \\
 s_2\left( D, C, \sigma_1 \right) &= \bar{S} - \delta_S \nonumber \\
 s_2\left( D, D, \sigma_1 \right) &= \bar{P} - \delta_P \nonumber \\
 s_2\left( C, C, \sigma_2 \right) &= \bar{R} + \delta_R \nonumber \\
 s_2\left( C, D, \sigma_2 \right) &= \bar{T} + \delta_T \nonumber \\
 s_2\left( D, C, \sigma_2 \right) &= \bar{S} + \delta_S \nonumber \\
 s_2\left( D, D, \sigma_2 \right) &= \bar{P} + \delta_P.
\end{align}
Although each stage game is not a symmetric game, the periodic prisoner's dilemma game is a symmetric game on average.
When $\delta_R=\delta_S=\delta_T=\delta_P=0$, the periodic prisoner's dilemma game is reduced to the repeated prisoner's dilemma game.

\section{Preliminaries}
\label{sec:preliminaries}
We introduce \emph{memory-one} strategies of player $j$ by $\left\{ T_j \left( a_j | \sigma, \bm{a}^\prime, \sigma^\prime \right) | a_j\in A_j, \sigma, \sigma^\prime \in \Sigma, \bm{a}^\prime \in \mathcal{A} \right\}$, where $T_j \left( a_j | \sigma, \bm{a}^\prime, \sigma^\prime \right)$ is the conditional probability to take $a_j$ when the state at the present round is $\sigma$, the previous action profile was $\bm{a}^\prime$, and the previous state was $\sigma^\prime$.
Generally, the joint probability distribution of the action profiles $\left\{ \bm{a}^{(t^\prime)} \right\}_{t^\prime=1}^t$ and the states $\left\{ \sigma^{(t^\prime)} \right\}_{t^\prime=1}^t$ is described as
\begin{align}
 & P\left( \left\{ \bm{a}^{(t^\prime)} \right\}_{t^\prime=1}^t, \left\{ \sigma^{(t^\prime)} \right\}_{t^\prime=1}^t \right) \nonumber \\
 &= \left[ \prod_{t^\prime=2}^t \left[ \prod_{k\in \mathcal{N}} T_k^{(t^\prime)} \left( a_k^{(t^\prime)} | \left\{ \bm{a}^{(t^{\prime\prime})} \right\}_{t^{\prime\prime}=1}^{t^\prime-1}, \left\{ \sigma^{(t^{\prime\prime})} \right\}_{t^{\prime\prime}=1}^{t^\prime} \right) \right] T_\mathrm{E} \left( \sigma^{(t^\prime)} | \bm{a}^{(t^\prime-1)}, \sigma^{(t^\prime-1)} \right) \right] \left[ \prod_{k\in \mathcal{N}} T_k^{(1)} \left( a_k^{(1)} | \sigma^{(1)} \right) \right] \nonumber \\
 &\quad \times P_\mathrm{E}^{(1)}\left( \sigma^{(1)} \right),
\end{align}
where $T_k^{(t^\prime)} \left( a_k^{(t^\prime)} | \left\{ \bm{a}^{(t^{\prime\prime})} \right\}_{t^{\prime\prime}=1}^{t^\prime-1}, \left\{ \sigma^{(t^{\prime\prime})} \right\}_{t^{\prime\prime}=1}^{t^\prime} \right)$ is the conditional probability to take $a_k^{(t^\prime)}$ in the $t^\prime$-th round when the history of the action profiles and the states is $\left\{ \bm{a}^{(t^{\prime\prime})} \right\}_{t^{\prime\prime}=1}^{t^\prime-1}$ and $\left\{ \sigma^{(t^{\prime\prime})} \right\}_{t^{\prime\prime}=1}^{t^\prime}$.
The joint probability distribution satisfies a recursion relation
\begin{align}
 & P\left( \left\{ \bm{a}^{(t^\prime)} \right\}_{t^\prime=1}^{t+1}, \left\{ \sigma^{(t^\prime)} \right\}_{t^\prime=1}^{t+1} \right) \nonumber \\
 &= \left[ \prod_{k\in \mathcal{N}} T_k^{(t+1)} \left( a_k^{(t+1)} | \left\{ \bm{a}^{(t^{\prime\prime})} \right\}_{t^{\prime\prime}=1}^{t}, \left\{ \sigma^{(t^{\prime\prime})} \right\}_{t^{\prime\prime}=1}^{t+1} \right) \right] T_\mathrm{E} \left( \sigma^{(t+1)} | \bm{a}^{(t)}, \sigma^{(t)} \right) P\left( \left\{ \bm{a}^{(t^\prime)} \right\}_{t^\prime=1}^t, \left\{ \sigma^{(t^\prime)} \right\}_{t^\prime=1}^t \right).
 \label{eq:recursion_joint}
\end{align}
We introduce the marginal distribution
\begin{align}
 P_t \left( \bm{a}^{(t)}, \sigma^{(t)} \right) &:= \sum_{\left\{ \bm{a}^{(t^\prime)} \right\}_{t^\prime=1}^{t-1}} \sum_{\left\{ \sigma^{(t^\prime)} \right\}_{t^\prime=1}^{t-1}} P\left( \left\{ \bm{a}^{(t^\prime)} \right\}_{t^\prime=1}^t, \left\{ \sigma^{(t^\prime)} \right\}_{t^\prime=1}^t \right),
\end{align}
and the time-averaged distribution
\begin{align}
 P^*\left( \bm{a}^\prime, \sigma^\prime \right) &:= \lim_{T\rightarrow \infty} \frac{1}{T} \sum_{t=1}^T P_t\left( \bm{a}^\prime, \sigma^\prime \right).
\end{align}
The following result is known as the Akin's lemma \cite{Aki2016} for stochastic games.
\begin{lemma}[\cite{MMHet2025}]
\label{lem:Akin}
If player $j$ uses a memory-one strategy $T_j$, it satisfies
\begin{align}
 0 &= \sum_{\bm{a}^\prime, \sigma^\prime} \left[ T_j \left( a_j | \sigma, \bm{a}^\prime, \sigma^\prime \right) T_\mathrm{E} \left( \sigma | \bm{a}^\prime, \sigma^\prime \right) - \delta_{a_j, a_j^\prime} \delta_{\sigma, \sigma^\prime} \right] P^*\left( \bm{a}^\prime, \sigma^\prime \right)
 \label{eq:Akin_SG}
\end{align}
for all $a_j$ and $\sigma$.
\end{lemma}
\begin{proof}
When we consider $\sum_{a_{-j}^{(t+1)}} \sum_{\left\{ \bm{a}^{(t^\prime)} \right\}_{t^\prime=1}^{t}} \sum_{\left\{ \sigma^{(t^\prime)} \right\}_{t^\prime=1}^{t}}$ of Eq. (\ref{eq:recursion_joint}), we obtain
\begin{align}
 & \sum_{\bm{a}^\prime, \sigma^\prime} \delta_{a_j^\prime, a_j^{(t+1)}} \delta_{\sigma^\prime, \sigma^{(t+1)}} P_{t+1} \left( \bm{a}^\prime, \sigma^\prime \right) \nonumber \\
 &= \sum_{\left\{ \bm{a}^{(t^\prime)} \right\}_{t^\prime=1}^{t}} \sum_{\left\{ \sigma^{(t^\prime)} \right\}_{t^\prime=1}^{t}} T_j^{(t+1)} \left( a_j^{(t+1)} | \left\{ \bm{a}^{(t^{\prime\prime})} \right\}_{t^{\prime\prime}=1}^{t}, \left\{ \sigma^{(t^{\prime\prime})} \right\}_{t^{\prime\prime}=1}^{t+1} \right) T_\mathrm{E} \left( \sigma^{(t+1)} | \bm{a}^{(t)}, \sigma^{(t)} \right) \nonumber \\
 &\quad \times P\left( \left\{ \bm{a}^{(t^\prime)} \right\}_{t^\prime=1}^t, \left\{ \sigma^{(t^\prime)} \right\}_{t^\prime=1}^t \right)
\end{align}
If player $j$ uses a memory-one strategy, we obtain
\begin{align}
 \sum_{\bm{a}^\prime, \sigma^\prime} \delta_{a_j^\prime, a_j^{(t+1)}} \delta_{\sigma^\prime, \sigma^{(t+1)}} P_{t+1} \left( \bm{a}^\prime, \sigma^\prime \right) &= \sum_{\bm{a}^\prime, \sigma^\prime} T_j \left( a_j^{(t+1)} | \sigma^{(t+1)}, \bm{a}^\prime, \sigma^\prime \right) T_\mathrm{E} \left( \sigma^{(t+1)} | \bm{a}^\prime, \sigma^\prime \right) P_t\left( \bm{a}^\prime, \sigma^\prime \right).
\end{align}
By renaming $a_j^{(t+1)} \rightarrow a_j$ and $\sigma^{(t+1)} \rightarrow \sigma$, and taking $\lim_{T\rightarrow \infty} \frac{1}{T} \sum_{t=1}^T$ of the both sides, we obtain
\begin{align}
 \sum_{\bm{a}^\prime, \sigma^\prime} \delta_{a_j^\prime, a_j} \delta_{\sigma^\prime, \sigma} P^* \left( \bm{a}^\prime, \sigma^\prime \right) &= \sum_{\bm{a}^\prime, \sigma^\prime} T_j \left( a_j | \sigma, \bm{a}^\prime, \sigma^\prime \right) T_\mathrm{E} \left( \sigma | \bm{a}^\prime, \sigma^\prime \right) P^*\left( \bm{a}^\prime, \sigma^\prime \right),
\end{align}
which is equivalent to Eq. (\ref{eq:Akin_SG}).
\end{proof}
Below we define
\begin{align}
 \hat{T}_j \left( a_j, \sigma | \bm{a}^\prime, \sigma^\prime \right) &:= T_j \left( a_j | \sigma, \bm{a}^\prime, \sigma^\prime \right) T_\mathrm{E} \left( \sigma | \bm{a}^\prime, \sigma^\prime \right) - \delta_{a_j, a_j^\prime} \delta_{\sigma, \sigma^\prime}
 \label{eq:PD_SG}
\end{align}
and call them the \emph{Press-Dyson vectors}.
We remark that the Press-Dyson vectors satisfy
\begin{align}
 \sum_{a_j, \sigma} \hat{T}_j \left( a_j, \sigma | \bm{a}^\prime, \sigma^\prime \right) &= 0
 \label{eq:PD_normalization}
\end{align}
for all $\bm{a}^\prime$ and $\sigma^\prime$.

A partial version of the Akin's lemma is also obtained from Lemma \ref{lem:Akin}.
\begin{lemma}
\label{lem:partialAkin}
If player $j$ uses a memory-one strategy $T_j$ which does not depend on the present state $\sigma$, it satisfies
\begin{align}
 0 &= \sum_{\bm{a}^\prime, \sigma^\prime} \left[ T_j \left( a_j | \bm{a}^\prime, \sigma^\prime \right) - \delta_{a_j, a_j^\prime} \right] P^*\left( \bm{a}^\prime, \sigma^\prime \right)
 \label{eq:partialAkin}
\end{align}
for all $a_j$.
\end{lemma}
\begin{proof}
If $T_j$ does not depend on the present state $\sigma$, Lemma \ref{lem:Akin} becomes
\begin{align}
 0 &= \sum_{\bm{a}^\prime, \sigma^\prime} \left[ T_j \left( a_j | \bm{a}^\prime, \sigma^\prime \right) T_\mathrm{E} \left( \sigma | \bm{a}^\prime, \sigma^\prime \right) - \delta_{a_j, a_j^\prime} \delta_{\sigma, \sigma^\prime} \right] P^*\left( \bm{a}^\prime, \sigma^\prime \right).
\end{align}
By taking $\sum_\sigma$ of the both sides, we obtain Eq. (\ref{eq:partialAkin}).
\end{proof}
Similarly as above, we define
\begin{align}
 \hat{T}_j \left( a_j | \bm{a}^\prime, \sigma^\prime \right) &:= T_j \left( a_j | \bm{a}^\prime, \sigma^\prime \right) - \delta_{a_j, a_j^\prime}
\end{align}
and call them the \emph{partial Press-Dyson vectors}.

We now introduce zero-determinant strategies in stochastic games.
We write $B\left( \bm{a}, \sigma \right):= \sum_{k\in \mathcal{N}} \alpha_k s_k \left( \bm{a}, \sigma \right) + \alpha_0$ with some coefficients $\{ \alpha_k \}$.
\begin{definition}[\cite{MMHet2025}]
\label{def:ZDS}
A memory-one strategy of player $j$ is a \emph{zero-determinant (ZD) strategy} controlling $B$ if it satisfies
\begin{align}
 \sum_{a_j, \sigma} c_{a_j, \sigma} \hat{T}_j \left( a_j, \sigma | \bm{a}^\prime, \sigma^\prime \right) &= B\left( \bm{a}^\prime, \sigma^\prime \right)
 \label{eq:ZDS}
\end{align}
with some coefficients $\left\{ c_{a_j, \sigma} \right\}$ and $B$ is not identically zero.
\end{definition}
As a direct consequence of Lemma \ref{lem:Akin}, the ZD strategy (\ref{eq:ZDS}) unilaterally enforces
\begin{align}
 0 &= \left\langle B \right\rangle^*,
\end{align}
where $\left\langle \cdot \right\rangle^*$ represents the expectation with respect to $P^*$.
It should be noted that $\mathcal{S}_k=\left\langle s_k \right\rangle^*$ for all $k\in \mathcal{N}$.

We can also construct ZD strategies by using Lemma \ref{lem:partialAkin}.
\begin{definition}
\label{def:partialZDS}
A memory-one strategy of player $j$ is a \emph{partial ZD strategy} controlling $B$ if it satisfies
\begin{align}
 \sum_{a_j} c_{a_j} \hat{T}_j \left( a_j | \bm{a}^\prime, \sigma^\prime \right) &= B\left( \bm{a}^\prime, \sigma^\prime \right)
 \label{eq:partialZDS}
\end{align}
with some coefficients $\left\{ c_{a_j} \right\}$ and $B$ is not identically zero.
\end{definition}
A partial ZD strategy (\ref{eq:partialZDS}) also unilaterally enforces
\begin{align}
 0 &= \left\langle B \right\rangle^*
\end{align}
as a result of Lemma \ref{lem:partialAkin}.

We also collectively write $\check{a}_j:=(a_j, \sigma)$.
A necessary condition for the existence of ZD strategies is given as follows.
\begin{proposition}
\label{prop:necessary_ZDS}
A ZD strategy of player $j$ controlling $B$ exists only if there are $\overline{\check{a}}_j, \underline{\check{a}}_j \in A_j\times \Sigma$ such that
\begin{align}
 B\left( \overline{\check{a}}_j, a_{-j} \right) &\geq 0 \quad \left( \forall a_{-j} \in A_{-j} \right) \nonumber \\
 B\left( \underline{\check{a}}_j, a_{-j} \right) &\leq 0 \quad \left( \forall a_{-j} \in A_{-j} \right).
 \label{eq:autocratic}
\end{align}
\end{proposition}
\begin{proof}
If a ZD strategy of player $j$ controlling $B$ exists, it satisfies Eq. (\ref{eq:ZDS}).
We introduce the notations
\begin{align}
 c_\mathrm{max} &:= \max_{\check{a}_j} c_{\check{a}_j} \nonumber \\
 c_\mathrm{min} &:= \min_{\check{a}_j} c_{\check{a}_j} 
\end{align}
and
\begin{align}
 \check{a}_{j, \mathrm{max}} &:= \arg \max_{\check{a}_j} c_{\check{a}_j} \nonumber \\
 \check{a}_{j, \mathrm{min}} &:= \arg \min_{\check{a}_j} c_{\check{a}_j}.
\end{align}
Due to Eq. (\ref{eq:PD_normalization}), Eq. (\ref{eq:ZDS}) is rewritten as
\begin{align}
 B\left( \bm{a}^\prime, \sigma^\prime \right) &= \sum_{\check{a}_j} \left( c_{\check{a}_j} - c_\mathrm{max} \right) \hat{T}_j \left( a_j, \sigma | \bm{a}^\prime, \sigma^\prime \right) \nonumber \\
 &= \sum_{\check{a}_j} \left( c_{\check{a}_j} - c_\mathrm{min} \right) \hat{T}_j \left( a_j, \sigma | \bm{a}^\prime, \sigma^\prime \right).
\end{align}
Then, because $\hat{T}_j \left( a_j, \sigma | \bm{a}^\prime, \sigma^\prime \right) = T_j \left( a_j | \sigma, \bm{a}^\prime, \sigma^\prime \right) T_\mathrm{E} \left( \sigma | \bm{a}^\prime, \sigma^\prime \right) \geq 0$ for $a_j\neq a_j^\prime$ or $\sigma\neq \sigma^\prime$, we find
\begin{align}
 B\left( \check{a}_{j, \mathrm{max}}, a_{-j}^\prime \right) &= \sum_{\check{a}_j} \left( c_{\check{a}_j} - c_\mathrm{max} \right) \hat{T}_j \left( a_j, \sigma | \check{a}_{j, \mathrm{max}}, a_{-j}^\prime \right) \nonumber \\
 &= \sum_{\check{a}_j \neq \check{a}_{j, \mathrm{max}}} \left( c_{\check{a}_j} - c_\mathrm{max} \right) \hat{T}_j \left( a_j, \sigma | \check{a}_{j, \mathrm{max}}, a_{-j}^\prime \right) \leq 0
\end{align}
and 
\begin{align}
 B\left( \check{a}_{j, \mathrm{min}}, a_{-j}^\prime \right) &= \sum_{\check{a}_j} \left( c_{\check{a}_j} - c_\mathrm{min} \right) \hat{T}_j \left( a_j, \sigma | \check{a}_{j, \mathrm{min}}, a_{-j}^\prime \right) \nonumber \\
 &= \sum_{\check{a}_j \neq \check{a}_{j, \mathrm{min}}} \left( c_{\check{a}_j} - c_\mathrm{min} \right) \hat{T}_j \left( a_j, \sigma | \check{a}_{j, \mathrm{min}}, a_{-j}^\prime \right) \geq 0.
\end{align}
Therefore, we can identify $\overline{\check{a}}_j=\check{a}_{j, \mathrm{min}}$ and $\underline{\check{a}}_j=\check{a}_{j, \mathrm{max}}$.
\end{proof}
It should be noted that the condition (\ref{eq:autocratic}) can be rewritten as \cite{Ued2026}
\begin{align}
 \max_{\check{a}_j} \min_{a_{-j}} B\left( \check{a}_j, a_{-j} \right) &\geq 0 \nonumber \\
 \min_{\check{a}_j} \max_{a_{-j}} B\left( \check{a}_j, a_{-j} \right) &\leq 0.
 \label{eq:autocratic_minimax}
\end{align}
Although the condition (\ref{eq:autocratic}) is a necessary condition for the existence of ZD strategies, it is not necessarily a sufficient condition.

In contrast, a necessary and sufficient condition for the existence of partial ZD strategies is given as follows.
\begin{proposition}
\label{prop:existence_partialZDS}
A partial ZD strategy of player $j$ controlling $B$ exists if and only if there are $\overline{a}_j, \underline{a}_j \in A_j$ such that
\begin{align}
 B\left( \overline{a}_j, a_{-j}, \sigma \right) &\geq 0 \quad \left( \forall a_{-j} \in A_{-j}, \forall \sigma \in \Sigma \right) \nonumber \\
 B\left( \underline{a}_j, a_{-j}, \sigma \right) &\leq 0 \quad \left( \forall a_{-j} \in A_{-j}, \forall \sigma \in \Sigma \right).
 \label{eq:autocratic_partial}
\end{align}
and $B$ is not identically zero.
\end{proposition}
\begin{proof}
For the case of partial ZD strategies, player $j$ uses a memory-one strategy $T_j$ which does not depend on the present environmental state.
Therefore, from the viewpoint of player $j$, an environment can be regarded as a player who takes action $\sigma$ simultaneously with player $j$.
This virtual change does not affect on the definition of partial ZD strategies (\ref{eq:partialZDS}), because they are fully determined by the strategy of player $j$ and the definition of $B$, both of which depend only on the ``previous action profile'' $(\bm{a}^\prime, \sigma^\prime)$.
Then, we can apply the results for repeated games \cite{Ued2022b}.
Particularly, the necessity is proved similarly as Proposition \ref{prop:necessary_ZDS}.
The sufficiency can be proved by explicitly constructing a ZD strategy \cite{Ued2022b}.
\end{proof}
Similarly as above, the condition (\ref{eq:autocratic_partial}) can be rewritten as
\begin{align}
 \max_{a_j} \min_{a_{-j}, \sigma} B\left( \bm{a}, \sigma \right) &\geq 0 \nonumber \\
 \min_{a_j} \max_{a_{-j}, \sigma} B\left( \bm{a}, \sigma \right) &\leq 0.
 \label{eq:autocratic_partial_minimax}
\end{align}
Because partial ZD strategies are ZD strategies, Proposition \ref{prop:existence_partialZDS} gives a sufficient condition for the existence of ZD strategies.
Although we consider a specific transition probability $T_\mathrm{E}$ in this paper, partial ZD strategies can be used against arbitrary $T_\mathrm{E}$.
ZD strategies of an environment \cite{DRWW2021,ZZDR2024} can also be constructed similarly as partial ZD strategies.

\section{Results}
\label{sec:results}
Here, we investigate the existence of \emph{fair} ZD strategies, which are ZD strategies controlling $s_1-s_2$, in the periodic prisoner's dilemma game.
It is convenient to introduce the notations
\begin{align}
 \bm{s}_j &:= \left(
\begin{array}{c}
s_j(C,C,\sigma_1) \\
s_j(C,D,\sigma_1) \\
s_j(D,C,\sigma_1) \\
s_j(D,D,\sigma_1) \\
s_j(C,C,\sigma_2) \\
s_j(C,D,\sigma_2) \\
s_j(D,C,\sigma_2) \\
s_j(D,D,\sigma_2)
\end{array}
\right)
\end{align}
and
\begin{align}
 \bm{\hat{T}}_j\left( a_j, \sigma \right) &:= \left(
\begin{array}{c}
\hat{T}_j(a_j, \sigma | C,C,\sigma_1) \\
\hat{T}_j(a_j, \sigma | C,D,\sigma_1) \\
\hat{T}_j(a_j, \sigma | D,C,\sigma_1) \\
\hat{T}_j(a_j, \sigma | D,D,\sigma_1) \\
\hat{T}_j(a_j, \sigma | C,C,\sigma_2) \\
\hat{T}_j(a_j, \sigma | C,D,\sigma_2) \\
\hat{T}_j(a_j, \sigma | D,C,\sigma_2) \\
\hat{T}_j(a_j, \sigma | D,D,\sigma_2)
\end{array}
\right).
\end{align}
We find that
\begin{align}
 \bm{B} &:= \bm{s}_1 - \bm{s}_2 = \left(
\begin{array}{c}
2\delta_R \\
\bar{S}-\bar{T}+\delta_S+\delta_T \\
\bar{T}-\bar{S}+\delta_S+\delta_T \\
2\delta_P \\
-2\delta_R \\
\bar{S}-\bar{T}-\delta_S-\delta_T \\
\bar{T}-\bar{S}-\delta_S-\delta_T \\
-2\delta_P
\end{array}
\right)
\end{align}
and
\begin{align}
 \bm{\hat{T}}_1\left( C, \sigma_1 \right) = \left(
\begin{array}{c}
-1 \\
-1 \\
0 \\
0 \\
T_1(C | \sigma_1,  C,C,\sigma_2) \\
T_1(C | \sigma_1,  C,D,\sigma_2) \\
T_1(C | \sigma_1,  D,C,\sigma_2) \\
T_1(C | \sigma_1,  D,D,\sigma_2)
\end{array}
\right), \quad
 \bm{\hat{T}}_1\left( D, \sigma_1 \right) = \left(
\begin{array}{c}
0 \\
0 \\
-1 \\
-1 \\
T_1(D | \sigma_1,  C,C,\sigma_2) \\
T_1(D | \sigma_1,  C,D,\sigma_2) \\
T_1(D | \sigma_1,  D,C,\sigma_2) \\
T_1(D | \sigma_1,  D,D,\sigma_2)
\end{array}
\right), \nonumber \\
\bm{\hat{T}}_1\left( C, \sigma_2 \right) = \left(
\begin{array}{c}
T_1(C | \sigma_2,  C,C,\sigma_1) \\
T_1(C | \sigma_2,  C,D,\sigma_1) \\
T_1(C | \sigma_2,  D,C,\sigma_1) \\
T_1(C | \sigma_2,  D,D,\sigma_1) \\
-1 \\
-1 \\
0 \\
0
\end{array}
\right), \quad
 \bm{\hat{T}}_1\left( D, \sigma_2 \right) = \left(
\begin{array}{c}
T_1(D | \sigma_2,  C,C,\sigma_1) \\
T_1(D | \sigma_2,  C,D,\sigma_1) \\
T_1(D | \sigma_2,  D,C,\sigma_1) \\
T_1(D | \sigma_2,  D,D,\sigma_1) \\
0 \\
0 \\
-1 \\
-1
\end{array}
\right).
\end{align}
By using this notation, a fair ZD strategy is one that satisfies
\begin{align}
 \sum_{a_j, \sigma} c_{a_j, \sigma} \bm{\hat{T}}_1\left( a_j, \sigma \right) &= \bm{B}
 \label{eq:fair_vec}
\end{align}
with some coefficients $\left\{ c_{a_j, \sigma} \right\}$.

\subsection{A necessary condition for the existence of fair ZD strategies}
First, we consider the consequence of Proposition \ref{prop:necessary_ZDS}.
\begin{theorem}
\label{thm:necessary_fair}
The necessary condition in Proposition \ref{prop:necessary_ZDS} for the existence of fair ZD strategies of player 1 is not satisfied if and only if
\begin{align}
 \delta_R>0, \delta_P>0, \delta_S+\delta_T< \bar{S}-\bar{T}
 \label{eq:nec_1}
\end{align}
or
\begin{align}
 \delta_R<0, \delta_P<0, \delta_S+\delta_T> \bar{T}-\bar{S}.
 \label{eq:nec_2}
\end{align}
\end{theorem}
\begin{proof}
First, we introduce $D_1:= \bar{S}-\bar{T}+\delta_S+\delta_T$ and $D_2:= \bar{T}-\bar{S}+\delta_S+\delta_T$, and write Eq. (\ref{eq:autocratic_minimax}) concretely:
\begin{align}
 \max\left\{ \min\left\{ 2\delta_R, D_1 \right\}, \min\left\{ D_2, 2\delta_P \right\}, \min\left\{ -2\delta_R, -D_2 \right\}, \min\left\{ -D_1, -2\delta_P \right\} \right\} &\geq 0 \nonumber \\
 \min\left\{ \max\left\{ 2\delta_R, D_1 \right\}, \max\left\{ D_2, 2\delta_P \right\}, \max\left\{ -2\delta_R, -D_2 \right\}, \max\left\{ -D_1, -2\delta_P \right\} \right\} &\leq 0.
\end{align}
These are further rewritten as
\begin{align}
 \max\left\{ \min\left\{ 2\delta_R, D_1 \right\}, \min\left\{ D_2, 2\delta_P \right\}, -\max\left\{ 2\delta_R, D_2 \right\}, -\max\left\{ D_1, 2\delta_P \right\} \right\} &\geq 0 \nonumber \\
 \min\left\{ \max\left\{ 2\delta_R, D_1 \right\}, \max\left\{ D_2, 2\delta_P \right\}, -\min\left\{ 2\delta_R, D_2 \right\}, -\min\left\{ D_1, 2\delta_P \right\} \right\} &\leq 0.
\end{align}
These inequalities are not satisfied if and only if
\begin{align}
 \min\left\{ 2\delta_R, D_1 \right\}<0, \min\left\{ D_2, 2\delta_P \right\}<0, \max\left\{ 2\delta_R, D_2 \right\}>0, \max\left\{ D_1, 2\delta_P \right\}>0
 \label{eq:fair_maxmin}
\end{align}
or
\begin{align}
 \max\left\{ 2\delta_R, D_1 \right\}>0, \max\left\{ D_2, 2\delta_P \right\}>0, \min\left\{ 2\delta_R, D_2 \right\}<0, \min\left\{ D_1, 2\delta_P \right\}<0.
 \label{eq:fair_minmax}
\end{align}
Taking $D_1<D_2$ into account, we consider 12 cases separately.
\begin{enumerate}
\item $D_1 < D_2 \leq 2\delta_R \leq 2\delta_P$\\
For this case, Eqs. (\ref{eq:fair_maxmin}) and (\ref{eq:fair_minmax}) are
\begin{align}
 D_1<0, D_2<0, 2\delta_R>0, 2\delta_P>0
\end{align}
and
\begin{align}
 2\delta_R>0, 2\delta_P>0, D_2<0, D_1<0,
\end{align}
respectively.

\item $D_1 \leq 2\delta_R < D_2 \leq 2\delta_P$\\
For this case, Eqs. (\ref{eq:fair_maxmin}) and (\ref{eq:fair_minmax}) are
\begin{align}
 D_1<0, D_2<0, D_2>0, 2\delta_P>0
\end{align}
and
\begin{align}
 2\delta_R>0, 2\delta_P>0, 2\delta_R<0, D_1<0,
\end{align}
respectively.
These inequalities are not satisfied.

\item $D_1 \leq 2\delta_R \leq 2\delta_P < D_2$\\
For this case, Eqs. (\ref{eq:fair_maxmin}) and (\ref{eq:fair_minmax}) are
\begin{align}
 D_1<0, 2\delta_P<0, D_2>0, 2\delta_P>0
\end{align}
and
\begin{align}
 2\delta_R>0, D_2>0, 2\delta_R<0, D_1<0,
\end{align}
respectively.
These inequalities are not satisfied.

\item $2\delta_R \leq D_1 < D_2 \leq 2\delta_P$\\
For this case, Eqs. (\ref{eq:fair_maxmin}) and (\ref{eq:fair_minmax}) are
\begin{align}
 2\delta_R<0, D_2<0, D_2>0, 2\delta_P>0
\end{align}
and
\begin{align}
 D_1>0, 2\delta_P>0, 2\delta_R<0, D_1<0,
\end{align}
respectively.
These inequalities are not satisfied.

\item $2\delta_R \leq D_1 \leq 2\delta_P < D_2$\\
For this case, Eqs. (\ref{eq:fair_maxmin}) and (\ref{eq:fair_minmax}) are
\begin{align}
 2\delta_R<0, 2\delta_P<0, D_2>0, 2\delta_P>0
\end{align}
and
\begin{align}
 D_1>0, D_2>0, 2\delta_R<0, D_1<0,
\end{align}
respectively.
These inequalities are not satisfied.

\item $2\delta_R \leq 2\delta_P \leq D_1 < D_2$\\
For this case, Eqs. (\ref{eq:fair_maxmin}) and (\ref{eq:fair_minmax}) are
\begin{align}
 2\delta_R<0, 2\delta_P<0, D_2>0, D_1>0
\end{align}
and
\begin{align}
 D_1>0, D_2>0, 2\delta_R<0, 2\delta_P<0,
\end{align}
respectively.

\item $D_1 < D_2 \leq 2\delta_P \leq 2\delta_R$\\
For this case, Eqs. (\ref{eq:fair_maxmin}) and (\ref{eq:fair_minmax}) are
\begin{align}
 D_1<0, D_2<0, 2\delta_R>0, 2\delta_P>0
\end{align}
and
\begin{align}
 2\delta_R>0, 2\delta_P>0, D_2<0, D_1<0,
\end{align}
respectively.

\item $D_1 \leq 2\delta_P < D_2 \leq 2\delta_R$\\
For this case, Eqs. (\ref{eq:fair_maxmin}) and (\ref{eq:fair_minmax}) are
\begin{align}
 D_1<0, 2\delta_P<0, 2\delta_R>0, 2\delta_P>0
\end{align}
and
\begin{align}
 2\delta_R>0, D_2>0, D_2<0, D_1<0,
\end{align}
respectively.
These inequalities are not satisfied.

\item $D_1 \leq 2\delta_P \leq 2\delta_R < D_2$\\
For this case, Eqs. (\ref{eq:fair_maxmin}) and (\ref{eq:fair_minmax}) are
\begin{align}
 D_1<0, 2\delta_P<0, D_2>0, 2\delta_P>0
\end{align}
and
\begin{align}
 2\delta_R>0, D_2>0, 2\delta_R<0, D_1<0,
\end{align}
respectively.
These inequalities are not satisfied.

\item $2\delta_P \leq D_1 < D_2 \leq 2\delta_R$\\
For this case, Eqs. (\ref{eq:fair_maxmin}) and (\ref{eq:fair_minmax}) are
\begin{align}
 D_1<0, 2\delta_P<0, 2\delta_R>0, D_1>0
\end{align}
and
\begin{align}
 2\delta_R>0, D_2>0, D_2<0, 2\delta_P<0,
\end{align}
respectively.
These inequalities are not satisfied.

\item $2\delta_P \leq D_1 \leq 2\delta_R < D_2$\\
For this case, Eqs. (\ref{eq:fair_maxmin}) and (\ref{eq:fair_minmax}) are
\begin{align}
 D_1<0, 2\delta_P<0, D_2>0, D_1>0
\end{align}
and
\begin{align}
 2\delta_R>0, D_2>0, 2\delta_R<0, 2\delta_P<0,
\end{align}
respectively.
These inequalities are not satisfied.

\item $2\delta_P \leq 2\delta_R \leq D_1 < D_2$\\
For this case, Eqs. (\ref{eq:fair_maxmin}) and (\ref{eq:fair_minmax}) are
\begin{align}
 2\delta_R<0, 2\delta_P<0, D_2>0, D_1>0
\end{align}
and
\begin{align}
 D_1>0, D_2>0, 2\delta_R<0, 2\delta_P<0,
\end{align}
respectively.
\end{enumerate}
Therefore, Eqs. (\ref{eq:fair_maxmin}) and (\ref{eq:fair_minmax}) are satisfied for cases 1, 6, 7, 12.
In other words, the necessary condition for the existence of fair ZD strategies is not satisfied if and only if
\begin{itemize}
 \item $D_1<D_2<0<2\delta_R \leq 2\delta_P$
 \item $2\delta_R \leq 2\delta_P <0<D_1<D_2$
 \item $D_1<D_2<0<2\delta_P \leq 2\delta_R$
 \item $2\delta_P \leq 2\delta_R <0<D_1<D_2$,
\end{itemize}
which is equivalent to the condition in Theorem \ref{thm:necessary_fair}.
\end{proof}

Thus, in contrast to the standard prisoner's dilemma game \cite{PreDys2012}, there are cases where no fair ZD strategies exist in the periodic prisoner's dilemma game.

\subsection{A necessary and sufficient condition for the existence of fair ZD strategies}
\label{subsec:necsuf}
Here, we provide a necessary and sufficient condition for the existence of fair ZD strategies.
It is useful to introduce a vector
\begin{align}
 \bm{e} &:= \left(
\begin{array}{c}
-1 \\
-1 \\
-1 \\
-1 \\
1 \\
1 \\
1 \\
1
\end{array}
\right),
\end{align}
because it satisfies
\begin{align}
\bm{\hat{T}}_1\left( D, \sigma_1 \right) &= \bm{e} - \bm{\hat{T}}_1\left( C, \sigma_1 \right) \nonumber \\
\bm{\hat{T}}_1\left( D, \sigma_2 \right) &= -\bm{e} - \bm{\hat{T}}_1\left( C, \sigma_2 \right).
\label{eq:e_PD}
\end{align}
\begin{theorem}
\label{thm:necsuf_fair}
Fair ZD strategies of player 1 exist if and only if
\begin{align}
 \left\{
\begin{array}{c}
\bar{S}-\bar{T}+\delta_S+\delta_T \leq 2\min\left\{ \delta_R, \delta_P \right\} \\
\bar{T}-\bar{S}+\delta_S+\delta_T \geq 2\max\left\{ \delta_R, \delta_P \right\}
\label{eq:necsuf_1}
\end{array}
\right.
\end{align}
or
\begin{align}
 \left\{
\begin{array}{c}
\bar{S}-\bar{T}+\delta_S+\delta_T \geq 2\min\left\{ \delta_R, \delta_P \right\} \\
\bar{T}-\bar{S}+\delta_S+\delta_T \leq 2\max\left\{ \delta_R, \delta_P \right\}.
\label{eq:necsuf_2}
\end{array}
\right.
\end{align}
\end{theorem}
Before proceeding to the proof, we explain the outline.
The basic strategy for the proof is to determine whether $\bm{s}_1 - \bm{s}_2$ can be represented as a feasible linear combination (\ref{eq:fair_vec}) of the Press-Dyson vectors $\bm{\hat{T}}_1\left( a_j, \sigma \right)$ under the probability constraints $0\leq T_1 \left( a_1 | \sigma, \bm{a}^\prime, \sigma^\prime \right) \leq 1$.
Therefore, we first find necessary conditions for the payoffs and coefficients of the Press-Dyson vectors under the constraints $0\leq T_1 \left( a_1 | \sigma, \bm{a}^\prime, \sigma^\prime \right) \leq 1$, and then we explicitly solve the simultaneous equations with respect to $T_1 \left( a_1 | \sigma, \bm{a}^\prime, \sigma^\prime \right)$ under the necessary conditions.
It should be noted that case distinctions with respect to the signs of coefficients of the Press-Dyson vectors are required.

\begin{proof}
According to Eqs. (\ref{eq:fair_vec}) and (\ref{eq:e_PD}), a fair ZD strategy, if exists, is written in the form
\begin{align}
 c_1 \bm{\hat{T}}_1\left( C, \sigma_1 \right) + c_2 \bm{\hat{T}}_1\left( C, \sigma_2 \right) + c_\mathrm{E} \bm{e} &= \bm{B},
\end{align}
where $c_1$, $c_2$ and $c_\mathrm{E}$ are some constants.
Explicitly, it is written as
\begin{align}
\left(
\begin{array}{c}
-c_1  + c_2 T_1(C | \sigma_2,  C,C,\sigma_1) -c_\mathrm{E} \\
-c_1  + c_2 T_1(C | \sigma_2,  C,D,\sigma_1) -c_\mathrm{E} \\
c_2 T_1(C | \sigma_2,  D,C,\sigma_1) -c_\mathrm{E} \\
c_2 T_1(C | \sigma_2,  D,D,\sigma_1) -c_\mathrm{E} \\
c_1 T_1(C | \sigma_1,  C,C,\sigma_2) -c_2 + c_\mathrm{E} \\
c_1 T_1(C | \sigma_1,  C,D,\sigma_2) -c_2 + c_\mathrm{E}  \\
c_1 T_1(C | \sigma_1,  D,C,\sigma_2) + c_\mathrm{E}  \\
c_1 T_1(C | \sigma_1,  D,D,\sigma_2) + c_\mathrm{E} 
\end{array}
\right)
 = \left(
\begin{array}{c}
2\delta_R \\
\bar{S}-\bar{T}+\delta_S+\delta_T \\
\bar{T}-\bar{S}+\delta_S+\delta_T \\
2\delta_P \\
-2\delta_R \\
\bar{S}-\bar{T}-\delta_S-\delta_T \\
\bar{T}-\bar{S}-\delta_S-\delta_T \\
-2\delta_P
\end{array}
\right).
\label{eq:fairZD_e}
\end{align}
According to the signs of $c_1$ and $c_2$, we consider four cases separately.

\begin{enumerate}
 \item $c_1\geq 0$ and $c_2\geq 0$\\
 For this case, the inequalities $0\leq T_1(C | \sigma, \bm{a}^\prime, \sigma^\prime) \leq 1$ and Eq. (\ref{eq:fairZD_e}) lead to
 \begin{align}
 -c_1-c_\mathrm{E} &\leq 2\delta_R \leq -c_1 + c_2 - c_\mathrm{E} \nonumber \\
 -c_1-c_\mathrm{E} &\leq \bar{S}-\bar{T}+\delta_S+\delta_T  \leq -c_1 + c_2 - c_\mathrm{E} \nonumber \\
 -c_\mathrm{E} &\leq \bar{T}-\bar{S}+\delta_S+\delta_T  \leq c_2-c_\mathrm{E} \nonumber \\
 -c_\mathrm{E} &\leq 2\delta_P  \leq c_2-c_\mathrm{E} \nonumber \\
 -c_2+c_\mathrm{E} &\leq -2\delta_R \leq c_1 - c_2 + c_\mathrm{E} \nonumber \\
 -c_2+c_\mathrm{E} &\leq \bar{S}-\bar{T}-\delta_S-\delta_T \leq c_1 - c_2 + c_\mathrm{E} \nonumber \\
 c_\mathrm{E} &\leq \bar{T}-\bar{S}-\delta_S-\delta_T \leq c_1+c_\mathrm{E} \nonumber \\
 c_\mathrm{E} &\leq -2\delta_P \leq c_1+c_\mathrm{E}.
\end{align}
These inequalities are equivalent to
 \begin{align}
 -c_1+c_2-c_\mathrm{E} &= 2\delta_R \nonumber \\
 -c_\mathrm{E} &= 2\delta_P \nonumber \\
 -c_1-c_\mathrm{E} &\leq \bar{S}-\bar{T}+\delta_S+\delta_T \leq 2\min\left\{ \delta_R, \delta_P \right\} \nonumber \\
 2\max\left\{ \delta_R, \delta_P \right\} &\leq \bar{T}-\bar{S}+\delta_S+\delta_T \leq c_2-c_\mathrm{E}.
\end{align}
Therefore, we find
\begin{align}
 -c_1+2\delta_P &\leq \bar{S}-\bar{T}+\delta_S+\delta_T \leq 2\min\left\{ \delta_R, \delta_P \right\} \leq 2\max\left\{ \delta_R, \delta_P \right\} \leq \bar{T}-\bar{S}+\delta_S+\delta_T \leq c_1+2\delta_R.
\end{align}
These inequalities imply Eq. (\ref{eq:necsuf_1}) and
\begin{align}
 c_1 &\geq \max\left\{  2\delta_P - (\bar{S}-\bar{T}+\delta_S+\delta_T), \bar{T}-\bar{S}+\delta_S+\delta_T - 2\delta_R \right\}.
 \label{eq:c1_g0g0}
\end{align}
It should be noted that $c_2=c_1+2\delta_R-2\delta_P$.
These are a necessary condition for the payoffs and $(c_1, c_2, c_\mathrm{E})$.

Next, we solve the simultaneous equations (\ref{eq:fairZD_e}) with respect to $T_1 \left( a_1 | \sigma, \bm{a}^\prime, \sigma^\prime \right)$ under the necessary condition.

If $c_1\neq 0$ and $c_2\neq 0$, from Eq. (\ref{eq:fairZD_e}) we find
\begin{align}
 T_1(C | \sigma_2,  C,C,\sigma_1) &= 1 \nonumber \\
 T_1(C | \sigma_2,  D,D,\sigma_1) &= 0 \nonumber \\
 T_1(C | \sigma_1,  C,C,\sigma_2) &= 1 \nonumber \\
 T_1(C | \sigma_1,  D,D,\sigma_2) &= 0
 \label{eq:fairZDS_g0g0_1}
\end{align}
and
\begin{align}
 T_1(C | \sigma_2,  C,D,\sigma_1) &= \frac{\bar{S}-\bar{T}+\delta_S+\delta_T + c_1 - 2\delta_P}{c_1+2\delta_R-2\delta_P} \nonumber \\
 T_1(C | \sigma_2,  D,C,\sigma_1) &= \frac{\bar{T}-\bar{S}+\delta_S+\delta_T - 2\delta_P}{c_1+2\delta_R-2\delta_P} \nonumber \\
 T_1(C | \sigma_1,  C,D,\sigma_2) &= \frac{\bar{S}-\bar{T}-\delta_S-\delta_T + c_1 + 2\delta_R}{c_1} \nonumber \\
 T_1(C | \sigma_1,  D,C,\sigma_2) &= \frac{\bar{T}-\bar{S}-\delta_S-\delta_T + 2\delta_P}{c_1}.
 \label{eq:fairZDS_g0g0_2}
\end{align}
Therefore, we can explicitly construct fair ZD strategies.

If $c_1 = 0$ and $c_2\neq 0$, from Eq. (\ref{eq:fairZD_e}) we find
\begin{align}
 -2\delta_R &= \bar{S}-\bar{T}-\delta_S-\delta_T = -c_2 +c_\mathrm{E} \nonumber \\
 -2\delta_P &= \bar{T}-\bar{S}-\delta_S-\delta_T = c_\mathrm{E}
 \label{eq:case_0p}
\end{align}
and
\begin{align}
 T_1(C | \sigma_2,  C,C,\sigma_1) &= \frac{2\delta_R + c_\mathrm{E}}{c_2} = 1 \nonumber \\
 T_1(C | \sigma_2,  C,D,\sigma_1) &= \frac{\bar{S}-\bar{T}+\delta_S+\delta_T+c_\mathrm{E}}{c_2} = 0 \nonumber \\
 T_1(C | \sigma_2,  D,C,\sigma_1) &= \frac{\bar{T}-\bar{S}+\delta_S+\delta_T+c_\mathrm{E}}{c_2} = 1 \nonumber \\
 T_1(C | \sigma_2,  D,D,\sigma_1) &= \frac{2\delta_P+c_\mathrm{E}}{c_2} = 0.
 \label{eq:fairZDS_g0g0_0p}
\end{align}
For this case, $T_1(C | \sigma_1,  \bm{a}^\prime, \sigma_2)$ is arbitrary.
Therefore, we can explicitly construct fair ZD strategies.
But the condition (\ref{eq:case_0p}) is a special case of Eq. (\ref{eq:necsuf_1}).

If $c_1 \neq 0$ and $c_2= 0$, from Eq. (\ref{eq:fairZD_e}) we find
\begin{align}
 2\delta_R &= \bar{S}-\bar{T}+\delta_S+\delta_T = -c_1 -c_\mathrm{E} \nonumber \\
 2\delta_P &= \bar{T}-\bar{S}+\delta_S+\delta_T = -c_\mathrm{E}
 \label{eq:case_p0}
\end{align}
and
\begin{align}
 T_1(C | \sigma_1,  C,C,\sigma_2) &= \frac{-2\delta_R - c_\mathrm{E}}{c_1} = 1 \nonumber \\
 T_1(C | \sigma_1,  C,D,\sigma_2) &= \frac{\bar{S}-\bar{T}-\delta_S-\delta_T-c_\mathrm{E}}{c_1} = 0 \nonumber \\
 T_1(C | \sigma_1,  D,C,\sigma_2) &= \frac{\bar{T}-\bar{S}-\delta_S-\delta_T-c_\mathrm{E}}{c_1} = 1 \nonumber \\
 T_1(C | \sigma_1,  D,D,\sigma_2) &= \frac{-2\delta_P-c_\mathrm{E}}{c_1} = 0.
 \label{eq:fairZDS_g0g0_p0}
\end{align}
For this case, $T_1(C | \sigma_2,  \bm{a}^\prime, \sigma_1)$ is arbitrary.
Therefore, we can explicitly construct fair ZD strategies.
But the condition (\ref{eq:case_p0}) is a special case of Eq. (\ref{eq:necsuf_1}).

Finally, if $c_1 = 0$ and $c_2= 0$, from Eq. (\ref{eq:fairZD_e}) we find
\begin{align}
 2\delta_R &= \bar{S}-\bar{T}+\delta_S+\delta_T = 2\delta_P = \bar{T}-\bar{S}+\delta_S+\delta_T = -c_\mathrm{E}
\end{align}
This contradicts with $\bar{T}-\bar{S} > 0$.
Therefore, we cannot construct a fair ZD strategy for the case.

 \item $c_1\geq 0$ and $c_2< 0$\\
 For this case, the inequalities $0\leq T_1(C | \sigma, \bm{a}^\prime, \sigma^\prime) \leq 1$ and Eq. (\ref{eq:fairZD_e}) lead to
 \begin{align}
 -c_1+c_2-c_\mathrm{E} &\leq 2\delta_R \leq -c_1 - c_\mathrm{E} \nonumber \\
 -c_1+c_2-c_\mathrm{E} &\leq \bar{S}-\bar{T}+\delta_S+\delta_T  \leq -c_1 - c_\mathrm{E} \nonumber \\
 c_2-c_\mathrm{E} &\leq \bar{T}-\bar{S}+\delta_S+\delta_T  \leq -c_\mathrm{E} \nonumber \\
 c_2-c_\mathrm{E} &\leq 2\delta_P  \leq -c_\mathrm{E} \nonumber \\
 -c_2+c_\mathrm{E} &\leq -2\delta_R \leq c_1 - c_2 + c_\mathrm{E} \nonumber \\
 -c_2+c_\mathrm{E} &\leq \bar{S}-\bar{T}-\delta_S-\delta_T \leq c_1 - c_2 + c_\mathrm{E} \nonumber \\
 c_\mathrm{E} &\leq \bar{T}-\bar{S}-\delta_S-\delta_T \leq c_1+c_\mathrm{E} \nonumber \\
 c_\mathrm{E} &\leq -2\delta_P \leq c_1+c_\mathrm{E}.
\end{align}
These inequalities are equivalent to
 \begin{align}
 -c_1-c_\mathrm{E} &= \bar{S}-\bar{T}+\delta_S+\delta_T \nonumber \\
 c_2-c_\mathrm{E} &= \bar{T}-\bar{S}+\delta_S+\delta_T \nonumber \\
 -c_1+c_2-c_\mathrm{E} &\leq 2\delta_R \leq \bar{S}-\bar{T}+\delta_S+\delta_T \nonumber \\
 \bar{T}-\bar{S}+\delta_S+\delta_T &\leq 2\delta_P \leq -c_\mathrm{E}.
\end{align}
Therefore, we find
\begin{align}
 \bar{T}-\bar{S}+\delta_S+\delta_T - c_1 &\leq 2\delta_R \leq \bar{S}-\bar{T}+\delta_S+\delta_T < \bar{T}-\bar{S}+\delta_S+\delta_T \leq 2\delta_P \leq \bar{S}-\bar{T}+\delta_S+\delta_T + c_1.
\end{align}
These inequalities imply
\begin{align}
\bar{S}-\bar{T}+\delta_S+\delta_T \geq 2\delta_R \nonumber \\
\bar{T}-\bar{S}+\delta_S+\delta_T \leq 2\delta_P
\label{eq:necsuf_2-1}
\end{align}
and
\begin{align}
 c_1 &\geq \max\left\{ \bar{T}-\bar{S}+\delta_S+\delta_T - 2\delta_R, 2\delta_P - (\bar{S}-\bar{T}+\delta_S+\delta_T) \right\}.
 \label{eq:c1_g0l0}
\end{align}
It should be noted that $c_2=2(\bar{T}-\bar{S})-c_1$.
These are a necessary condition for the payoffs and $(c_1, c_2, c_\mathrm{E})$.

Next, we solve the simultaneous equations (\ref{eq:fairZD_e}) with respect to $T_1 \left( a_1 | \sigma, \bm{a}^\prime, \sigma^\prime \right)$ under the necessary condition.

If $c_1\neq 0$, from Eq. (\ref{eq:fairZD_e}) we find
\begin{align}
 T_1(C | \sigma_2,  C,D,\sigma_1) &= 0 \nonumber \\
 T_1(C | \sigma_2,  D,C,\sigma_1) &= 1 \nonumber \\
 T_1(C | \sigma_1,  C,D,\sigma_2) &= 0 \nonumber \\
 T_1(C | \sigma_1,  D,C,\sigma_2) &= 1
 \label{eq:fairZDS_g0l0_1}
\end{align}
and
\begin{align}
 T_1(C | \sigma_2,  C,C,\sigma_1) &= \frac{\bar{S}-\bar{T}+\delta_S+\delta_T - 2\delta_R}{c_1-2(\bar{T}-\bar{S})} \nonumber \\
 T_1(C | \sigma_2,  D,D,\sigma_1) &= \frac{\bar{S}-\bar{T}+\delta_S+\delta_T + c_1 - 2\delta_P}{c_1-2(\bar{T}-\bar{S})} \nonumber \\
 T_1(C | \sigma_1,  C,C,\sigma_2) &= \frac{\bar{T}-\bar{S}+\delta_S+\delta_T - 2\delta_R}{c_1} \nonumber \\
 T_1(C | \sigma_1,  D,D,\sigma_2) &= \frac{\bar{S}-\bar{T}+\delta_S+\delta_T + c_1 - 2\delta_P}{c_1}.
 \label{eq:fairZDS_g0l0_2}
\end{align}
Therefore, we can explicitly construct fair ZD strategies.

If $c_1 = 0$, from Eq. (\ref{eq:fairZD_e}) we find
\begin{align}
 -2\delta_R &= \bar{S}-\bar{T}-\delta_S-\delta_T = -c_2 +c_\mathrm{E} \nonumber \\
 -2\delta_P &= \bar{T}-\bar{S}-\delta_S-\delta_T = c_\mathrm{E}.
\end{align}
Then we obtain
\begin{align}
 0< -c_2 = -2 (\bar{T}-\bar{S})<0,
\end{align}
leading to contradiction.
Therefore, we cannot construct a fair ZD strategy for the case.

 \item $c_1< 0$ and $c_2\geq 0$\\
 For this case, the inequalities $0\leq T_1(C | \sigma, \bm{a}^\prime, \sigma^\prime) \leq 1$ and Eq. (\ref{eq:fairZD_e}) lead to
 \begin{align}
 -c_1-c_\mathrm{E} &\leq 2\delta_R \leq -c_1 + c_2 - c_\mathrm{E} \nonumber \\
 -c_1-c_\mathrm{E} &\leq \bar{S}-\bar{T}+\delta_S+\delta_T  \leq -c_1 + c_2 - c_\mathrm{E} \nonumber \\
 -c_\mathrm{E} &\leq \bar{T}-\bar{S}+\delta_S+\delta_T  \leq c_2 -c_\mathrm{E} \nonumber \\
 -c_\mathrm{E} &\leq 2\delta_P  \leq c_2 -c_\mathrm{E} \nonumber \\
 c_1-c_2+c_\mathrm{E} &\leq -2\delta_R \leq - c_2 + c_\mathrm{E} \nonumber \\
 c_1-c_2+c_\mathrm{E} &\leq \bar{S}-\bar{T}-\delta_S-\delta_T \leq - c_2 + c_\mathrm{E} \nonumber \\
 c_1+c_\mathrm{E} &\leq \bar{T}-\bar{S}-\delta_S-\delta_T \leq c_\mathrm{E} \nonumber \\
 c_1+c_\mathrm{E} &\leq -2\delta_P \leq c_\mathrm{E}.
\end{align}
These inequalities are equivalent to
 \begin{align}
 -c_1-c_\mathrm{E} &= \bar{S}-\bar{T}+\delta_S+\delta_T \nonumber \\
 c_2-c_\mathrm{E} &= \bar{T}-\bar{S}+\delta_S+\delta_T \nonumber \\
 \bar{T}-\bar{S}+\delta_S+\delta_T &\leq 2\delta_R \leq -c_1+c_2-c_\mathrm{E} \nonumber \\
 -c_\mathrm{E} &\leq 2\delta_P \leq \bar{S}-\bar{T}+\delta_S+\delta_T.
\end{align}
Therefore, we find
\begin{align}
 \bar{S}-\bar{T}+\delta_S+\delta_T + c_1 &\leq 2\delta_P \leq \bar{S}-\bar{T}+\delta_S+\delta_T < \bar{T}-\bar{S}+\delta_S+\delta_T \leq 2\delta_R \leq \bar{T}-\bar{S}+\delta_S+\delta_T - c_1.
\end{align}
These inequalities imply
\begin{align}
\bar{S}-\bar{T}+\delta_S+\delta_T \geq 2\delta_P \nonumber \\
\bar{T}-\bar{S}+\delta_S+\delta_T \leq 2\delta_R
\label{eq:necsuf_2-2}
\end{align}
and
\begin{align}
 c_1 &\leq \min\left\{ \bar{T}-\bar{S}+\delta_S+\delta_T - 2\delta_R, 2\delta_P - (\bar{S}-\bar{T}+\delta_S+\delta_T) \right\}.
 \label{eq:c1_l0g0}
\end{align}
It should be noted that $c_2=2(\bar{T}-\bar{S})-c_1$.
These are a necessary condition for the payoffs and $(c_1, c_2, c_\mathrm{E})$.

Next, we solve the simultaneous equations (\ref{eq:fairZD_e}) with respect to $T_1 \left( a_1 | \sigma, \bm{a}^\prime, \sigma^\prime \right)$ under the necessary condition.

If $c_2\neq 0$, from Eq. (\ref{eq:fairZD_e}) we find
\begin{align}
 T_1(C | \sigma_2,  C,D,\sigma_1) &= 0 \nonumber \\
 T_1(C | \sigma_2,  D,C,\sigma_1) &= 1 \nonumber \\
 T_1(C | \sigma_1,  C,D,\sigma_2) &= 0 \nonumber \\
 T_1(C | \sigma_1,  D,C,\sigma_2) &= 1
 \label{eq:fairZDS_l0g0_1}
\end{align}
and
\begin{align}
 T_1(C | \sigma_2,  C,C,\sigma_1) &= \frac{2\delta_R- (\bar{S}-\bar{T}+\delta_S+\delta_T)}{2(\bar{T}-\bar{S})-c_1} \nonumber \\
 T_1(C | \sigma_2,  D,D,\sigma_1) &= \frac{2\delta_P - c_1 - (\bar{S}-\bar{T}+\delta_S+\delta_T)}{2(\bar{T}-\bar{S})-c_1} \nonumber \\
 T_1(C | \sigma_1,  C,C,\sigma_2) &= \frac{\bar{T}-\bar{S}+\delta_S+\delta_T - 2\delta_R}{c_1} \nonumber \\
 T_1(C | \sigma_1,  D,D,\sigma_2) &= \frac{\bar{S}-\bar{T}+\delta_S+\delta_T + c_1 - 2\delta_P}{c_1}.
 \label{eq:fairZDS_l0g0_2}
\end{align}
Therefore, we can explicitly construct fair ZD strategies.

If $c_2 = 0$, from Eq. (\ref{eq:fairZD_e}) we find
\begin{align}
 2\delta_R &= \bar{S}-\bar{T}+\delta_S+\delta_T = -c_1 -c_\mathrm{E} \nonumber \\
 2\delta_P &= \bar{T}-\bar{S}+\delta_S+\delta_T = -c_\mathrm{E}.
\end{align}
Then we obtain
\begin{align}
 0< -c_1 = -2 (\bar{T}-\bar{S})<0,
\end{align}
leading to contradiction.
Therefore, we cannot construct a fair ZD strategy for the case.

\item $c_1< 0$ and $c_2< 0$\\
 For this case, the inequalities $0\leq T_1(C | \sigma, \bm{a}^\prime, \sigma^\prime) \leq 1$ and Eq. (\ref{eq:fairZD_e}) lead to
 \begin{align}
 -c_1+c_2-c_\mathrm{E} &\leq 2\delta_R \leq -c_1 - c_\mathrm{E} \nonumber \\
 -c_1+c_2-c_\mathrm{E} &\leq \bar{S}-\bar{T}+\delta_S+\delta_T  \leq -c_1 - c_\mathrm{E} \nonumber \\
 c_2-c_\mathrm{E} &\leq \bar{T}-\bar{S}+\delta_S+\delta_T  \leq -c_\mathrm{E} \nonumber \\
 c_2-c_\mathrm{E} &\leq 2\delta_P  \leq -c_\mathrm{E} \nonumber \\
 c_1-c_2+c_\mathrm{E} &\leq -2\delta_R \leq - c_2 + c_\mathrm{E} \nonumber \\
 c_1-c_2+c_\mathrm{E} &\leq \bar{S}-\bar{T}-\delta_S-\delta_T \leq - c_2 + c_\mathrm{E} \nonumber \\
 c_1+c_\mathrm{E} &\leq \bar{T}-\bar{S}-\delta_S-\delta_T \leq c_\mathrm{E} \nonumber \\
 c_1+c_\mathrm{E} &\leq -2\delta_P \leq c_\mathrm{E}.
\end{align}
These inequalities are equivalent to
 \begin{align}
 -c_1+c_2-c_\mathrm{E} &= 2\delta_R \nonumber \\
 -c_\mathrm{E} &=2\delta_P \nonumber \\
 2\max\left\{ \delta_R, \delta_P \right\} &\leq \bar{S}-\bar{T}+\delta_S+\delta_T \leq -c_1-c_\mathrm{E} \nonumber \\
 c_2-c_\mathrm{E} &\leq \bar{T}-\bar{S}+\delta_S+\delta_T \leq 2\min\left\{ \delta_R, \delta_P \right\}.
\end{align}
Therefore, we find
\begin{align}
 c_1+2\delta_R &\leq \bar{T}-\bar{S}+\delta_S+\delta_T \leq 2\min\left\{ \delta_R, \delta_P \right\} \leq 2\max\left\{ \delta_R, \delta_P \right\} \leq \bar{S}-\bar{T}+\delta_S+\delta_T \leq -c_1+2\delta_P.
\end{align}
However, these inequalities contradict with $\bar{T}-\bar{S}>0$.
Therefore, it is impossible to construct a fair ZD strategy for this case.
\end{enumerate}

According to the four cases, we find that fair ZD strategies exist if and only if the condition (\ref{eq:necsuf_1}), (\ref{eq:necsuf_2-1}), or (\ref{eq:necsuf_2-2}) holds.
(For each case, there exists feasible $(c_1, c_2, c_\mathrm{E})$.)
It should be noted that the conditions (\ref{eq:necsuf_2-1}) and (\ref{eq:necsuf_2-2}) are integrated into Eq. (\ref{eq:necsuf_2}).
\end{proof}

As a result, even if the necessary condition in Theorem \ref{thm:necessary_fair} is satisfied, it is not always possible to construct a fair ZD strategy.
For example, if
\begin{align}
 \delta_R &> 0 \nonumber \\
 \delta_P &< 0 \nonumber \\
 \bar{S}-\bar{T}+\delta_S+\delta_T &< 2\delta_P \nonumber \\
 \bar{T}-\bar{S}+\delta_S+\delta_T &< 2\delta_R
\end{align}
hold, the necessary condition in Theorem \ref{thm:necessary_fair} is satisfied but the condition in Theorem \ref{thm:necsuf_fair} is not satisfied.

In addition, if Eq. (\ref{eq:nec_1}) holds, it actually satisfies neither Eq. (\ref{eq:necsuf_1}) nor Eq. (\ref{eq:necsuf_2}).
We obtain the same result for Eq. (\ref{eq:nec_2}).
By using the notations in the proof of Theorem \ref{thm:necessary_fair}, fair ZD strategies exist for cases 3, 4, 9, 10.

In order to check the validity of this Theorem, we have performed numerical simulations.
In Figure \ref{fig:linear_ZDS}, we have provided our numerical results.
In the left panel, we display the result for the case $\left( \bar{R}, \bar{S}, \bar{T}, \bar{P} \right)=(3, 0, 5, 1)$ and $\left( \delta_{R}, \delta_{S}, \delta_{T}, \delta_{P} \right)=(1, 1, 1, 1)$.
In the right panel, we display the result for the case $\left( \bar{R}, \bar{S}, \bar{T}, \bar{P} \right)=(3, 0, 5, 1)$ and $\left( \delta_{R}, \delta_{S}, \delta_{T}, \delta_{P} \right)=(2, 1, 1, 1)$.
For both cases, the condition (\ref{eq:necsuf_1}) holds, and a fair ZD strategy is given by (\ref{eq:fairZDS_g0g0_1}) and (\ref{eq:fairZDS_g0g0_2}), where $c_1$ is given by the equality condition of Eq. (\ref{eq:c1_g0g0}).
The strategy of player $2$ is given by $200$ randomly generated memory-one strategies.
Each $\left\langle s_j \right\rangle^{*}$ is calculated by time average over $10^6$ time steps.
We find that the fair ZD strategy indeed enforces a linear relation $\left\langle s_1 \right\rangle^{*}=\left\langle s_2 \right\rangle^{*}$.
Numerical results for other parameter values are given in \ref{app:numerical}.
\begin{figure}
\begin{center}
\includegraphics[clip, width=6.0cm]{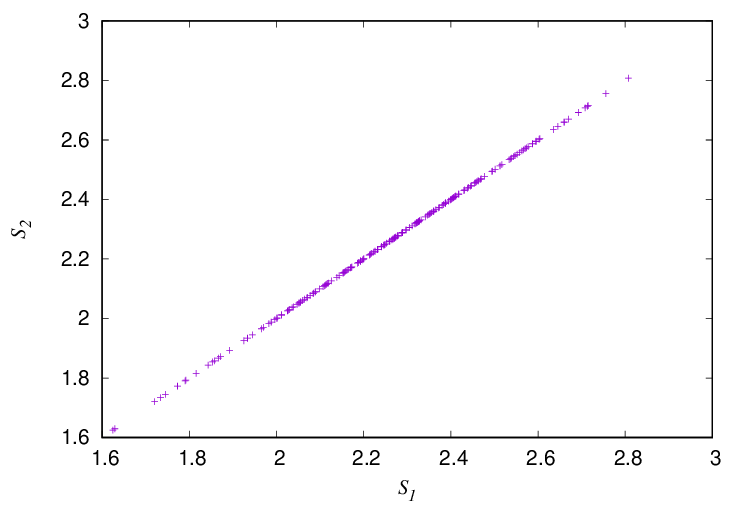}
\includegraphics[clip, width=6.0cm]{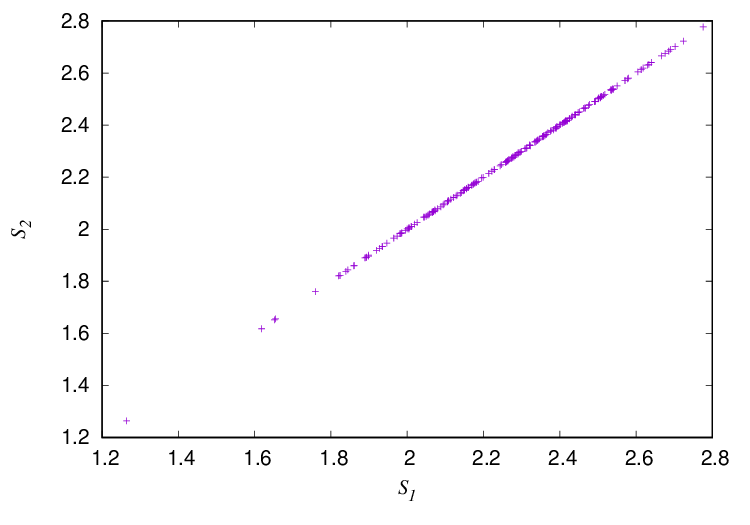}
\end{center}
\caption{Relations between $\left\langle s_1 \right\rangle^{*}$ and $\left\langle s_2 \right\rangle^{*}$ when player $1$ uses the fair ZD strategy (\ref{eq:fairZDS_g0g0_1}) and (\ref{eq:fairZDS_g0g0_2}) with $c_1$ satisfying the equality condition of Eq. (\ref{eq:c1_g0g0}) and player $2$ uses $200$ randomly generated memory-one strategies. The payoffs in the left panel are set to $\left( \bar{R}, \bar{S}, \bar{T}, \bar{P} \right)=(3, 0, 5, 1)$ and $\left( \delta_{R}, \delta_{S}, \delta_{T}, \delta_{P} \right)=(1, 1, 1, 1)$ (satisfying $\delta_R + \delta_P = \delta_S + \delta_T$). The payoffs in the right panel are set to $\left( \bar{R}, \bar{S}, \bar{T}, \bar{P} \right)=(3, 0, 5, 1)$ and $\left( \delta_{R}, \delta_{S}, \delta_{T}, \delta_{P} \right)=(2, 1, 1, 1)$ (not satisfying $\delta_R + \delta_P = \delta_S + \delta_T$). Each $\left\langle s_j \right\rangle^{*}$ is calculated by time average over $10^6$ time steps. The fair ZD strategy indeed enforces a linear relation $\left\langle s_1 \right\rangle^{*}=\left\langle s_2 \right\rangle^{*}$.}
\label{fig:linear_ZDS}
\end{figure}

An intuitive explanation of Theorem \ref{thm:necsuf_fair} is as follows.
Based on $\bm{B}$, we can construct a relative payoff game \cite{DOS2012a} as in Tables \ref{table:relPD1} and \ref{table:relPD2}, where payoffs are given by the payoff difference between two players.
A relative payoff game is a zero-sum game.
Now we investigate the meaning of the conditions (\ref{eq:necsuf_1}) and (\ref{eq:necsuf_2}) in Theorem \ref{thm:necsuf_fair}.
It should be noted that the condition (\ref{eq:necsuf_1}) is equivalent to one of the following two inequalities:
\begin{align}
 && \bar{T}-\bar{S}+\delta_S+\delta_T \geq 2\delta_R \geq 2\delta_P \geq \bar{S}-\bar{T}+\delta_S+\delta_T \nonumber \\
 && \bar{T}-\bar{S}+\delta_S+\delta_T \geq 2\delta_P \geq 2\delta_R \geq \bar{S}-\bar{T}+\delta_S+\delta_T.
\end{align}
Similarly, the condition (\ref{eq:necsuf_2}) is equivalent to one of the following two inequalities:
\begin{align}
 && 2\delta_R \geq \bar{T}-\bar{S}+\delta_S+\delta_T \geq \bar{S}-\bar{T}+\delta_S+\delta_T \geq 2\delta_P \nonumber \\
 && 2\delta_P \geq \bar{T}-\bar{S}+\delta_S+\delta_T \geq \bar{S}-\bar{T}+\delta_S+\delta_T \geq 2\delta_R.
\end{align}
\begin{enumerate}
 \item $\bar{T}-\bar{S}+\delta_S+\delta_T \geq 2\delta_R \geq 2\delta_P \geq \bar{S}-\bar{T}+\delta_S+\delta_T$\\
 For this case, we can find that $D$ dominates $C$ in both states $\sigma_1$ and $\sigma_2$ in the relative payoff game.
 \item $\bar{T}-\bar{S}+\delta_S+\delta_T \geq 2\delta_P \geq 2\delta_R \geq \bar{S}-\bar{T}+\delta_S+\delta_T$\\
 For this case, we can find that $D$ dominates $C$ in both states $\sigma_1$ and $\sigma_2$ in the relative payoff game.
 \item $2\delta_R \geq \bar{T}-\bar{S}+\delta_S+\delta_T \geq \bar{S}-\bar{T}+\delta_S+\delta_T \geq 2\delta_P$\\
 For this case, we can find that $C$ dominates $D$ in state $\sigma_1$, and $D$ dominates $C$ in state $\sigma_2$ in the relative payoff game.
 \item $2\delta_P \geq \bar{T}-\bar{S}+\delta_S+\delta_T \geq \bar{S}-\bar{T}+\delta_S+\delta_T \geq 2\delta_R$\\
 For this case, we can find that $D$ dominates $C$ in state $\sigma_1$, and $C$ dominates $D$ in state $\sigma_2$ in the relative payoff game.
\end{enumerate}
Therefore, under the conditions (\ref{eq:necsuf_1}) and (\ref{eq:necsuf_2}), dominant actions exist in both states in the relative payoff game.
A fair ZD strategy may enforce $\left\langle s_1 \right\rangle^{*}=\left\langle s_2 \right\rangle^{*}$ by switching the dominant action (obtaining large relative payoffs) and the dominated action (obtaining small relative payoffs) appropriately.
\begin{table}[tb]
  \centering
  \caption{Relative payoffs in state $\sigma_1$.}
  \begin{tabular}{|c|cc|} \hline
    & $C$ & $D$ \\ \hline
   $C$ & $2\delta_R, -2\delta_R$ & $\bar{S}-\bar{T}+\delta_S+\delta_T, \bar{T}-\bar{S}-\delta_S-\delta_T$ \\
   $D$ & $\bar{T}-\bar{S}+\delta_S+\delta_T, \bar{S}-\bar{T}-\delta_S-\delta_T$ & $2\delta_P, -2\delta_P$ \\ \hline
  \end{tabular}
  \label{table:relPD1}
  
  \centering
  \caption{Relative payoffs in state $\sigma_2$.}
  \begin{tabular}{|c|cc|} \hline
    & $C$ & $D$ \\ \hline
   $C$ & $-2\delta_R, 2\delta_R$ & $\bar{S}-\bar{T}-\delta_S-\delta_T, \bar{T}-\bar{S}+\delta_S+\delta_T$ \\
   $D$ & $\bar{T}-\bar{S}-\delta_S-\delta_T, \bar{S}-\bar{T}+\delta_S+\delta_T$ & $-2\delta_P, 2\delta_P$ \\ \hline
  \end{tabular}
  \label{table:relPD2}
\end{table}

\subsection{Relation between fair ZD strategies and the Tit-for-Tat strategy}
Next, we investigate a relation between fair ZD strategies and the Tit-for-Tat strategy.
The \emph{Tit-for-Tat} (TFT)  is a memory-one strategy which imitates the previous action of the opponent \cite{RCO1965,AxeHam1981}.
It was known that TFT is a fair ZD strategy in the standard prisoner's dilemma game \cite{PreDys2012}.
\begin{theorem}
\label{thm:TFT_fair}
TFT is a fair ZD strategy if and only if
\begin{align}
 \delta_R + \delta_P &= \delta_S + \delta_T.
 \label{eq:condition_TFT_fair}
\end{align}
\end{theorem}
\begin{proof}
It should be noted that the Press-Dyson vectors for TFT of player 1 are written as
\begin{align}
 \bm{\hat{T}}_1\left( C, \sigma_1 \right) = \left(
\begin{array}{c}
-1 \\
-1 \\
0 \\
0 \\
1 \\
0 \\
1 \\
0
\end{array}
\right), \quad
 \bm{\hat{T}}_1\left( D, \sigma_1 \right) = \left(
\begin{array}{c}
0 \\
0 \\
-1 \\
-1 \\
0 \\
1 \\
0 \\
1
\end{array}
\right), \quad
\bm{\hat{T}}_1\left( C, \sigma_2 \right) = \left(
\begin{array}{c}
1 \\
0 \\
1 \\
0 \\
-1 \\
-1 \\
0 \\
0
\end{array}
\right), \quad
 \bm{\hat{T}}_1\left( D, \sigma_2 \right) = \left(
\begin{array}{c}
0 \\
1 \\
0 \\
1 \\
0 \\
0 \\
-1 \\
-1
\end{array}
\right).
\end{align}
TFT is a fair ZD strategy if and only if there exist coefficients $\left\{ c_{a_j, \sigma} \right\}$ such that
\begin{align}
 \sum_{a_j, \sigma} c_{a_j, \sigma} \bm{\hat{T}}_1\left( a_j, \sigma \right) &= \bm{s}_1 - \bm{s}_2,
\end{align}
that is,
\begin{align}
\left(
\begin{array}{c}
-c_{C, \sigma_1} + c_{C, \sigma_2} \\
-c_{C, \sigma_1} + c_{D, \sigma_2} \\
-c_{D, \sigma_1} + c_{C, \sigma_2} \\
-c_{D, \sigma_1} + c_{D, \sigma_2} \\
c_{C, \sigma_1} - c_{C, \sigma_2}  \\
c_{D, \sigma_1} - c_{C, \sigma_2}  \\
c_{C, \sigma_1} - c_{D, \sigma_2}  \\
c_{D, \sigma_1} - c_{D, \sigma_2} 
\end{array}
\right) = \left(
\begin{array}{c}
2\delta_R \\
\bar{S}-\bar{T}+\delta_S+\delta_T \\
\bar{T}-\bar{S}+\delta_S+\delta_T \\
2\delta_P \\
-2\delta_R \\
\bar{S}-\bar{T}-\delta_S-\delta_T \\
\bar{T}-\bar{S}-\delta_S-\delta_T \\
-2\delta_P
\end{array}
\right).
\end{align}
Then, we find that
\begin{align}
 c_{C, \sigma_2} - c_{D, \sigma_2} &= 2\delta_R - \left( \bar{S}-\bar{T}+\delta_S+\delta_T  \right)
\end{align}
and
\begin{align}
 c_{C, \sigma_2} - c_{D, \sigma_2} &= \bar{T}-\bar{S}+\delta_S+\delta_T - 2\delta_P.
\end{align}
Therefore, these two quantities must be equal to each other, which is equivalent to Eq. (\ref{eq:condition_TFT_fair}).
It should be noted that this condition means that unilateral deviation from $(C,C)$ by the opponent can be compensated in any cycles.
That is, both the cycle $(C, C, \sigma_1) \rightarrow (C, D, \sigma_2) \rightarrow (D, C, \sigma_1) \rightarrow (C, C, \sigma_2) \rightarrow \cdots$ and the cycle $(C, C, \sigma_1) \rightarrow (C, D, \sigma_2) \rightarrow (D, D, \sigma_1) \rightarrow (D, C, \sigma_2) \rightarrow \cdots$ must result in zero total relative payoffs; See the last paragraph of this subsection.

Conversely, if the condition (\ref{eq:condition_TFT_fair}) holds, we find
\begin{align}
 \bm{s}_1 - \bm{s}_2 &= \left(
\begin{array}{c}
2\delta_R \\
\bar{S}-\bar{T}+\delta_R+\delta_P \\
\bar{T}-\bar{S}+\delta_R+\delta_P \\
2\delta_P \\
-2\delta_R \\
\bar{S}-\bar{T}-\delta_R-\delta_P \\
\bar{T}-\bar{S}-\delta_R-\delta_P \\
-2\delta_P
\end{array}
\right) = \left(
\begin{array}{c}
2\delta_R \\
\bar{S}-\bar{T}+\delta_R-\delta_P+2\delta_P \\
\bar{T}-\bar{S}-\delta_R+\delta_P+2\delta_R \\
2\delta_P \\
-2\delta_R \\
\bar{S}-\bar{T}+\delta_R-\delta_P -2\delta_R \\
\bar{T}-\bar{S}-\delta_R+\delta_P -2\delta_P \\
-2\delta_P
\end{array}
\right) \nonumber \\
  &= 2\delta_R \bm{\hat{T}}_1\left( C, \sigma_2 \right) + 2\delta_P \bm{\hat{T}}_1\left( D, \sigma_2 \right) + \left( \bar{T}-\bar{S}-\delta_R+\delta_P \right) \left[ \bm{\hat{T}}_1\left( C, \sigma_1 \right) + \bm{\hat{T}}_1\left( C, \sigma_2 \right) \right],
\end{align}
which means that TFT is a fair ZD strategy.
\end{proof}

Again, in contrast to the standard prisoner's dilemma game \cite{PreDys2012}, TFT is not necessarily a fair ZD strategy in the periodic prisoner's dilemma game.
The condition (\ref{eq:condition_TFT_fair}) can be rewritten as
\begin{align}
 R^{(1)} + P^{(1)} - T^{(1)} - S^{(1)} &= R^{(2)} + P^{(2)} - T^{(2)} - S^{(2)}.
\end{align}
This condition implies that some baseline in state $\sigma_1$ is equivalent to that in state $\sigma_2$.
We remark that this condition is equivalent to the condition that the asymmetric prisoner's dilemma game becomes a potential game \cite{Ued2025}.
Originally, in Ref. \cite{MMHet2025}, the case $\left( R^{(1)}, S^{(1)}, T^{(1)}, P^{(1)} \right) = \left( b, (b-c)/2, b, (b-c)/2 \right)$ and $\left( R^{(2)}, S^{(2)}, T^{(2)}, P^{(2)} \right) = \left( 0, 0, b/2, b/2 \right)$ with $b>0$ and $c>0$ was investigated.
For such parameters, the condition (\ref{eq:condition_TFT_fair}) is satisfied, and therefore TFT is a fair ZD strategy.

In order to check the validity of this Theorem, we have performed numerical simulations.
In Figure \ref{fig:linear_TFT}, we have provided our numerical results, with the same parameter values as Figure \ref{fig:linear_ZDS}.
In the left panel, we display the result for the case $\left( \bar{R}, \bar{S}, \bar{T}, \bar{P} \right)=(3, 0, 5, 1)$ and $\left( \delta_{R}, \delta_{S}, \delta_{T}, \delta_{P} \right)=(1, 1, 1, 1)$, which satisfies the condition (\ref{eq:condition_TFT_fair}).
In the right panel, we display the result for the case $\left( \bar{R}, \bar{S}, \bar{T}, \bar{P} \right)=(3, 0, 5, 1)$ and $\left( \delta_{R}, \delta_{S}, \delta_{T}, \delta_{P} \right)=(2, 1, 1, 1)$, which does not satisfy the condition (\ref{eq:condition_TFT_fair}).
The strategy of player $2$ is given by $200$ randomly generated memory-one strategies.
Each $\left\langle s_j \right\rangle^{*}$ is calculated by time average over $10^6$ time steps.
We find that TFT enforces a linear relation $\left\langle s_1 \right\rangle^{*}=\left\langle s_2 \right\rangle^{*}$ in the left panel, whereas TFT does not enforce a linear relation $\left\langle s_1 \right\rangle^{*}=\left\langle s_2 \right\rangle^{*}$ in the right panel.
This result is consistent with Theorem \ref{thm:TFT_fair}.
\begin{figure}
\begin{center}
\includegraphics[clip, width=6.0cm]{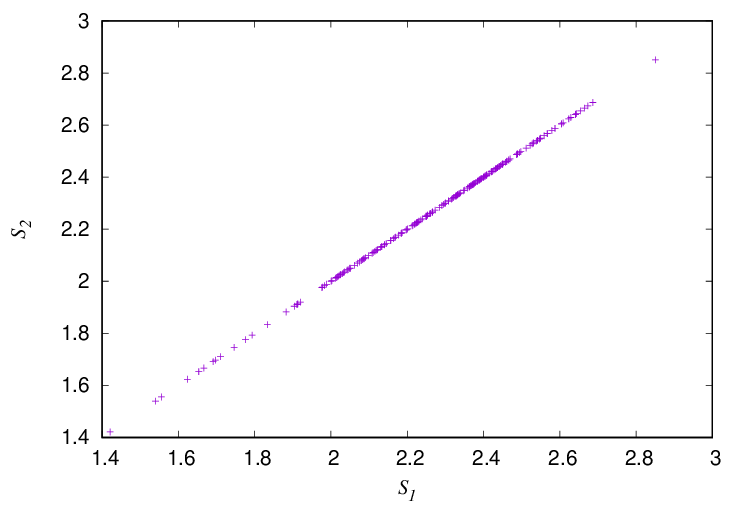}
\includegraphics[clip, width=6.0cm]{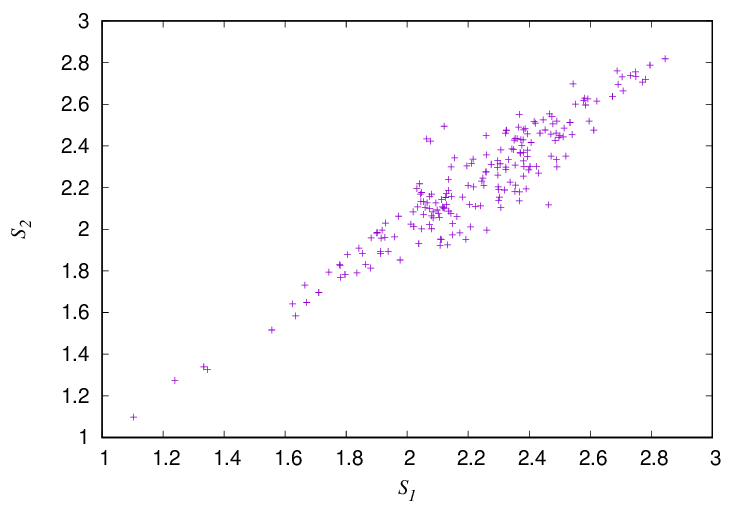}
\end{center}
\caption{Relations between $\left\langle s_1 \right\rangle^{*}$ and $\left\langle s_2 \right\rangle^{*}$ when player $1$ uses TFT and player $2$ uses $200$ randomly generated memory-one strategies. The payoffs in the left panel are set to $\left( \bar{R}, \bar{S}, \bar{T}, \bar{P} \right)=(3, 0, 5, 1)$ and $\left( \delta_{R}, \delta_{S}, \delta_{T}, \delta_{P} \right)=(1, 1, 1, 1)$ (satisfying $\delta_R + \delta_P = \delta_S + \delta_T$). The payoffs in the right panel are set to $\left( \bar{R}, \bar{S}, \bar{T}, \bar{P} \right)=(3, 0, 5, 1)$ and $\left( \delta_{R}, \delta_{S}, \delta_{T}, \delta_{P} \right)=(2, 1, 1, 1)$ (not satisfying $\delta_R + \delta_P = \delta_S + \delta_T$). Each $\left\langle s_j \right\rangle^{*}$ is calculated by time average over $10^6$ time steps. TFT enforces a linear relation $\left\langle s_1 \right\rangle^{*}=\left\langle s_2 \right\rangle^{*}$ in the left panel, whereas TFT does not enforce a linear relation $\left\langle s_1 \right\rangle^{*}=\left\langle s_2 \right\rangle^{*}$ in the right panel.}
\label{fig:linear_TFT}
\end{figure}

An intuitive explanation of Theorem \ref{thm:TFT_fair} is as follows.
Suppose that player $1$ adopts TFT.
In order for TFT to be a fair ZD strategy, TFT must not be exploited by the opponent in any cycle.
For example, if we consider a cycle
\begin{align}
 (C, C, \sigma_1) & \rightarrow (C, C, \sigma_2) \rightarrow \cdots,
\end{align}
the total payoff difference of this cycle is
\begin{align}
 2\delta_R + \left( -2\delta_R \right) &= 0.
\end{align}
Similarly, if we consider a cycle
\begin{align}
 (C, C, \sigma_1) & \rightarrow (C, D, \sigma_2) \rightarrow (D, C, \sigma_1) \rightarrow (C, C, \sigma_2) \rightarrow \cdots,
\end{align}
the total payoff difference of this cycle is
\begin{align}
 2\delta_R + \left( \bar{S}-\bar{T}-\delta_S-\delta_T \right) + \left( \bar{T}-\bar{S}+\delta_S+\delta_T \right) + \left( -2\delta_R \right) &= 0.
\end{align}
Now, let us consider a cycle
\begin{align}
 (C, C, \sigma_1) & \rightarrow (C, D, \sigma_2) \rightarrow (D, D, \sigma_1) \rightarrow (D, C, \sigma_2) \rightarrow \cdots.
\end{align}
The total payoff difference of this cycle is
\begin{align}
 2\delta_R + \left( \bar{S}-\bar{T}-\delta_S-\delta_T \right) + 2\delta_P + \left( \bar{T}-\bar{S}-\delta_S-\delta_T \right) &= 2 \left( \delta_R + \delta_P - \delta_S - \delta_T \right).
\end{align}
If TFT is a fair ZD strategy, TFT must enforce $\left\langle s_1 \right\rangle^* = \left\langle s_2 \right\rangle^*$ even if this cycle is infinitely repeated.
Therefore, we conclude that $\delta_R + \delta_P - \delta_S - \delta_T=0$ must hold.

\subsection{On the existence of fair partial ZD strategies}
Here we consider the consequence of Proposition \ref{prop:existence_partialZDS}.
\begin{theorem}
\label{thm:fair_partial}
The necessary and sufficient condition in Proposition \ref{prop:existence_partialZDS} for the existence of fair partial ZD strategies of player 1 is equivalent to
\begin{align}
 \delta_R &= 0 \nonumber \\
 \delta_P &= 0 \nonumber \\
 \left| \delta_S+\delta_T \right| &\leq \bar{T}-\bar{S}.
 \label{eq:fair_partial}
\end{align}
\end{theorem}
\begin{proof}
The necessary and sufficient condition for the existence of fair partial ZD strategies of player 1 is Eq. (\ref{eq:autocratic_partial_minimax}).
By introducing $D_1:= \bar{S}-\bar{T}+\delta_S+\delta_T$ and $D_2:= \bar{T}-\bar{S}+\delta_S+\delta_T$, these conditions are explicitly written as
\begin{align}
 \max\left\{ \min\left\{ 2\delta_R, D_1, -2\delta_R, -D_2 \right\}, \min\left\{ D_2, 2\delta_P, -D_1, -2\delta_P \right\} \right\} &\geq 0 \nonumber \\
 \min\left\{ \max\left\{ 2\delta_R, D_1, -2\delta_R, -D_2 \right\}, \max\left\{ D_2, 2\delta_P, -D_1, -2\delta_P \right\} \right\} &\leq 0. 
 \label{eq:fairpartial_minimax}
\end{align}

If $\delta_R\neq 0$ and $\delta_P\neq 0$, these inequalities cannot be satisfied, because
\begin{align}
 \min\left\{ 2\delta_R, D_1, -2\delta_R, -D_2 \right\} &\leq -2\left| \delta_R \right| <0 \nonumber \\
 \min\left\{ D_2, 2\delta_P, -D_1, -2\delta_P \right\} &\leq -2\left| \delta_P \right| <0,
\end{align}
for example.

If $\delta_R= 0$ and $\delta_P\neq 0$, the first inequality in Eq. (\ref{eq:fairpartial_minimax}) becomes
\begin{align}
 \max\left\{ \min\left\{ 0, D_1, -D_2 \right\}, \min\left\{ D_2, 2\delta_P, -D_1, -2\delta_P \right\} \right\} &\geq 0.
\end{align}
Due to
\begin{align}
 \min\left\{ D_2, 2\delta_P, -D_1, -2\delta_P \right\} &\leq -2\left| \delta_P \right| <0,
\end{align}
the inequality holds if and only if $D_1\geq 0$ and $-D_2\geq 0$, that is,
\begin{align}
 \bar{S}-\bar{T}+\delta_S+\delta_T &\geq 0 \nonumber \\
 \bar{S}-\bar{T}-\delta_S-\delta_T &\geq 0.
\end{align}
Then we obtain $\bar{S}-\bar{T}\geq 0$, leading to contradiction.

If $\delta_R\neq 0$ and $\delta_P= 0$, the second inequality in Eq. (\ref{eq:fairpartial_minimax}) becomes
\begin{align}
 \min\left\{ \max\left\{ 2\delta_R, D_1, -2\delta_R, -D_2 \right\}, \max\left\{ D_2, 0, -D_1 \right\} \right\} &\leq 0. 
\end{align}
Due to
\begin{align}
 \max\left\{ 2\delta_R, D_1, -2\delta_R, -D_2 \right\} &\geq 2\left| \delta_R \right| >0,
\end{align}
the inequality holds if and only if $D_2\leq 0$ and $-D_1\leq 0$, that is,
\begin{align}
 \bar{T}-\bar{S}+\delta_S+\delta_T &\leq 0 \nonumber \\
 \bar{T}-\bar{S}-\delta_S-\delta_T &\leq 0.
\end{align}
Then we obtain $\bar{T}-\bar{S}\leq 0$, leading to contradiction.

If $\delta_R= 0$ and $\delta_P= 0$, the inequalities (\ref{eq:fairpartial_minimax}) become
\begin{align}
 \max\left\{ \min\left\{ 0, D_1, -D_2 \right\}, \min\left\{ D_2, 0, -D_1 \right\} \right\} &\geq 0 \nonumber \\
 \min\left\{ \max\left\{ 0, D_1, -D_2 \right\}, \max\left\{ D_2, 0, -D_1 \right\} \right\} &\leq 0,
\end{align}
that is, 
\begin{align}
 \max\left\{ \min\left\{ 0, \bar{S}-\bar{T}-\left| \delta_S+\delta_T \right| \right\}, \min\left\{ 0, \bar{T}-\bar{S}-\left| \delta_S+\delta_T \right| \right\} \right\} &\geq 0 \nonumber \\
 \min\left\{ \max\left\{ 0, \bar{S}-\bar{T}+\left| \delta_S+\delta_T \right| \right\}, \max\left\{ 0, \bar{T}-\bar{S}+\left| \delta_S+\delta_T \right| \right\} \right\} &\leq 0,
\end{align}
The first inequality holds if and only if
\begin{align}
 \bar{T}-\bar{S}-\left| \delta_S+\delta_T \right| &\geq 0.
\end{align}
The second inequality holds if and only if
\begin{align}
 \bar{S}-\bar{T}+\left| \delta_S+\delta_T \right| &\leq 0,
\end{align}
It should be noted that these two inequalities are equivalent.
Therefore, a fair partial ZD strategy exists if and only if the condition (\ref{eq:fair_partial}) is satisfied.
Indeed, we can define $\overline{a}_1=D$ and $\underline{a}_1=C$ for this case.
\end{proof}

We remark that the condition in Theorem \ref{thm:fair_partial} is not necessarily contained in Eq. (\ref{eq:condition_TFT_fair}).
Therefore, if $\delta_R = 0$, $\delta_P = 0$, and $\bar{T}-\bar{S}\geq \left| \delta_S+\delta_T \right| \neq 0$, there exist fair ZD strategies which are not TFT.
In contrast, if $\delta_R \neq 0$ and $\delta_P \neq 0$ with $\delta_R + \delta_P = \delta_S + \delta_T$, there exist fair ZD strategies which are not fair partial ZD strategies.

\section{Concluding remarks}
\label{sec:conclusion}
In this paper, we made two main contributions.
First, we specified a necessary and sufficient condition for the existence of fair ZD strategies in the periodic prisoner's dilemma game (Theorem \ref{thm:necsuf_fair}).
We found that this existence condition is quite different from a necessary condition in Theorem \ref{thm:necessary_fair}, which is direct consequence of Proposition \ref{prop:necessary_ZDS}.
In repeated games, the necessary condition in Proposition \ref{prop:necessary_ZDS} is also a sufficient condition for the existence of ZD strategies \cite{Ued2022b}.
Therefore, this result highlights difference between ZD strategies in repeated games and those in stochastic games.

Second, we also specified the relation between TFT and fair ZD strategies (Theorem \ref{thm:TFT_fair}).
In the standard repeated prisoner's dilemma game, TFT is always a fair ZD strategy \cite{PreDys2012}.
However, in the periodic prisoner's dilemma game, this equivalence holds only in special cases.
In other words, simple imitation of the opponent is not always unbeatable \cite{DOS2014,Ued2022}.
This result also characterizes complexity of stochastic games.
The results of this paper are summarized in Table \ref{tbl:summary}.
\begin{table}
\caption{Summary of this study.}
\begin{tabular}{|c|l|}
  \hline 
  The condition for the existence of fair ZD strategies & $\left\{
\begin{array}{c}
\bar{S}-\bar{T}+\delta_S+\delta_T \leq 2\min\left\{ \delta_R, \delta_P \right\} \\
\bar{T}-\bar{S}+\delta_S+\delta_T \geq 2\max\left\{ \delta_R, \delta_P \right\}
\end{array}
\right.$\\
   & or\\
   & $\left\{
\begin{array}{c}
\bar{S}-\bar{T}+\delta_S+\delta_T \geq 2\min\left\{ \delta_R, \delta_P \right\} \\
\bar{T}-\bar{S}+\delta_S+\delta_T \leq 2\max\left\{ \delta_R, \delta_P \right\}
\end{array}
\right.$ \\
  \hline
  The condition under which TFT is a fair ZD strategy & $\delta_R + \delta_P = \delta_S + \delta_T$\\
  \hline
  The condition for the existence of fair partial ZD strategies & $\delta_R = 0$, $\delta_P = 0$, $\left| \delta_S+\delta_T \right| \leq \bar{T}-\bar{S}$ \\
  \hline 
\end{tabular}
\label{tbl:summary}
\end{table}

Discrepancy between the necessary condition in Proposition \ref{prop:necessary_ZDS} and the existence condition of ZD strategies in stochastic games is similar to discrepancy found in repeated games with discounting \cite{HTS2015,McAHau2016} or discrepancy found in repeated games with imperfect monitoring \cite{HRZ2015,MamIch2020}.
In these situations, we cannot also choose Press-Dyson vectors arbitrarily.
In repeated games with discounting, Press-Dyson vectors are restricted by a discount factor.
In repeated games with imperfect monitoring, Press-Dyson vectors are restricted by imperfect observation.
Similarly, in stochastic games, Press-Dyson vectors are restricted by the transition probability of an environmental state; See Eq. (\ref{eq:PD_SG}), where players cannot choose $T_\mathrm{E}$.
Finding general conditions for the existence of ZD strategies under such restrictions is a significant open problem.

Although we do not know the situations which are exactly described by the periodic prisoner's dilemma game, this game can be regarded as a toy model where two roles (such as offense and defense) of players alternate deterministically.
Examples of such situations include baseball and role-playing video games, where the ``offense'' role and the ``defense'' role are alternately played.
We believe that the periodic prisoner's dilemma game is one of the simplest models of such situations.
We also remark that the periodic prisoner's dilemma game in this paper is different from alternating games \cite{McAHau2017,PNH2022}, where players alternately update their actions.
Our results show that, even if both players can play two roles equally, it may be impossible for one player to unilaterally equalize the payoffs of two players.

The periodic prisoner's dilemma game is different from a stochastic prisoner's dilemma game which has recently been substantially investigated \cite{HSCN2018,LiuWu2022,WZWL2026,ZYRW2026}.
Although the latter is also a two-state stochastic game, two states describe better and worse environmental states, respectively, and typically, transition to a worse state is coupled to defection.
The typical payoff matrices are given in Tables \ref{table:PD1_Hilbe} and \ref{table:PD2_Hilbe}, with $T_n > R_n > P_n > S_n$ for $n=1, 2$.
For such a game, fair partial ZD strategies always exist, as a direct consequence of Proposition \ref{prop:existence_partialZDS} with $\overline{a}_j=D$ and $\underline{a}_j=C$.
This is because each stage game is a symmetric game similarly as the standard prisoner's dilemma game.
In contrast, in the periodic prisoner's dilemma game, the roles of two states are symmetric but each stage game is not always a symmetric game (Tables \ref{table:PD1} and \ref{table:PD2}), and fair partial ZD strategies do not always exist.
Therefore, even if a stochastic game is a two-state two-player two-action game, the existence condition of ZD strategies is quite different.
\begin{table}[tb]
  \centering
  \caption{Payoffs in state $\sigma_1$.}
  \begin{tabular}{|c|cc|} \hline
    & $C$ & $D$ \\ \hline
   $C$ & $R_1, R_1$ & $S_1, T_1$ \\
   $D$ & $T_1, S_1$ & $P_1, P_1$ \\ \hline
  \end{tabular}
  \label{table:PD1_Hilbe}
  
  \centering
  \caption{Payoffs in state $\sigma_2$.}
  \begin{tabular}{|c|cc|} \hline
    & $C$ & $D$ \\ \hline
   $C$ & $R_2, R_2$ & $S_2, T_2$ \\
   $D$ & $T_2, S_2$ & $P_2, P_2$ \\ \hline
  \end{tabular}
  \label{table:PD2_Hilbe}
\end{table}

In this paper, we focused on only fair ZD strategies.
In the repeated prisoner's dilemma game, there are other ZD strategies, such as the equalizer strategies \cite{BNS1997}, the extortionate strategies \cite{PreDys2012}, and the generous strategies \cite{StePlo2013}.
The equalizer strategies unilaterally set the opponent's payoff.
The extortionate strategies unilaterally obtain the payoff not less than that of the opponent.
The generous strategies unilaterally obtain the payoff not more than that of the opponent but promote mutual cooperation.
We expect that the existence condition of other ZD strategies in the periodic prisoner's dilemma game is also more complicated than that in the repeated prisoner's dilemma game.
In a special case, we can show that the equalizer strategy exists as in \ref{app:equalizer}.
The existence condition of other ZD strategies in the periodic prisoner's dilemma game should be investigated in future.

Finally, we remark on the size of memory of ZD strategies.
In this paper, we consider only memory-one ZD strategies.
Since there are two states in the periodic prisoner's dilemma game, memory-$m$ ZD strategies with $m\geq 2$ \cite{Ued2021b} may be useful to control payoffs in the game.
In the repeated prisoner's dilemma game ($\delta_R=\delta_S=\delta_T=\delta_P=0$), TFT is a fair ZD strategy because it controls the cumulated payoff difference between two players within $\bar{T}-\bar{S}$.
Concretely, if player $1$ adopts TFT and we consider the transition $(C,C) \rightarrow (C,D) \rightarrow (D,D)$, the cumulated payoff difference is $\bar{S}-\bar{T}$, and this cannot decrease furthermore.
However, in the periodic prisoner's dilemma game, a memory-two strategy which imitates the opponent's action before last (memory-two TFT) may not have such a property due to finiteness of $\delta$.
For example, when the condition (\ref{eq:nec_1}) in Theorem \ref{thm:necessary_fair} holds, player $1$ loses (wins) for $(C,D)$ and $(D,C)$ and wins (loses) for $(C,C)$ and $(D,D)$ in state $\sigma_1$ ($\sigma_2$); See Tables \ref{table:relPD1} and \ref{table:relPD2}.
This is because each stage game is not necessarily a symmetric game, in contrast to the prisoner's dilemma game.
Therefore, winning or losing depends on actions of both players similarly as the matching pennies games, and such memory-two strategy cannot probably control the cumulated payoff difference.
Nevertheless, it is also expected that results on multichannel games  \cite{DHNH2020,Ued2026}, where multiple repeated games are simultaneously played, can be applied to periodic games, by regarding each channel as each state.
Analysis of the existence of fair memory-$m$ ZD strategies with $m\geq 2$ remains to be solved.

\section*{Acknowledgement}
The authors thank the anonymous reviewers for their constructive comments.
This study was supported by Toyota Riken Scholar Program and JSPS KAKENHI Grant Number JP26K21335.

\appendix
\setcounter{figure}{0}
\def\thesubsection{\Alph{section}.\arabic{subsection}}
\section{Additional numerical results}
\label{app:numerical}
In Section \ref{subsec:necsuf}, we provided numerical results for the case where the condition (\ref{eq:necsuf_1}) holds.
In this appendix, we provide additional numerical results for other cases.

\subsection{Boundary case}
According to the proof of Theorem \ref{thm:necsuf_fair}, if the equalities
\begin{align}
 2\delta_R &= \bar{T}-\bar{S}+\delta_S+\delta_T \nonumber \\
 2\delta_P &= \bar{S}-\bar{T}+\delta_S+\delta_T
 \label{eq:boundary1}
\end{align}
hold in the condition (\ref{eq:necsuf_1}), a fair ZD strategy is given by Eq. (\ref{eq:fairZDS_g0g0_0p}) with arbitrary $T_1(C | \sigma_1,  \bm{a}^\prime, \sigma_2)$.
In order to check the validity, we perform numerical simulation for the case $\left( \bar{R}, \bar{S}, \bar{T}, \bar{P} \right)=(3, 0, 5, 1)$ and $\left( \delta_{R}, \delta_{S}, \delta_{T}, \delta_{P} \right)=(3.5, 1, 1, -1.5)$ (satisfying Eq. (\ref{eq:boundary1})).
We set $T_1(C | \sigma_1,  \bm{a}^\prime, \sigma_2)=0$ for all $\bm{a}^\prime$.
The result is given in Figure \ref{fig:linear_ZDS_boundary}.
This result is consistent with Theorem \ref{thm:necsuf_fair}.
\begin{figure}
\begin{center}
\includegraphics[clip, width=6.0cm]{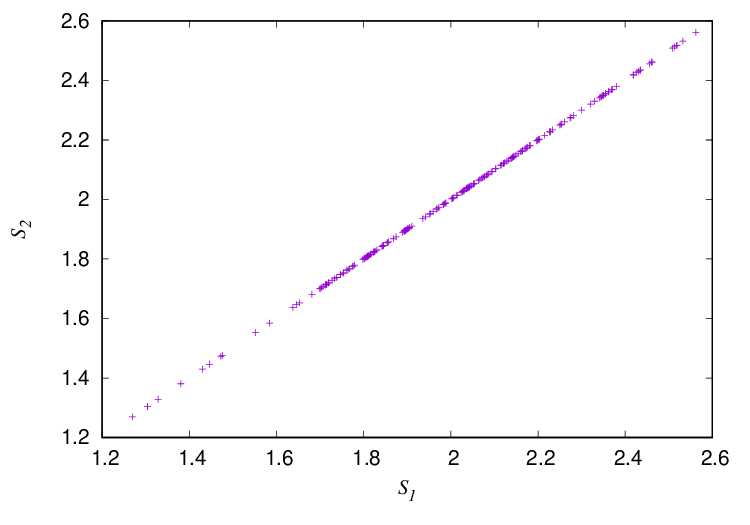}
\end{center}
\caption{A linear relation between $\left\langle s_1 \right\rangle^{*}$ and $\left\langle s_2 \right\rangle^{*}$ when player $1$ uses the fair ZD strategy (\ref{eq:fairZDS_g0g0_0p}) and $T_1(C | \sigma_1,  \bm{a}^\prime, \sigma_2)=0$ and player $2$ uses $200$ randomly generated memory-one strategies. The payoffs are set to $\left( \bar{R}, \bar{S}, \bar{T}, \bar{P} \right)=(3, 0, 5, 1)$ and $\left( \delta_{R}, \delta_{S}, \delta_{T}, \delta_{P} \right)=(3.5, 1, 1, -1.5)$. Each $\left\langle s_j \right\rangle^{*}$ is calculated by time average over $10^6$ time steps. The fair ZD strategy indeed enforces a linear relation $\left\langle s_1 \right\rangle^{*}=\left\langle s_2 \right\rangle^{*}$.}
\label{fig:linear_ZDS_boundary}
\end{figure}

\subsection{The case (\ref{eq:necsuf_2-1})}
For this case, a fair ZD strategy is given by Eqs. (\ref{eq:fairZDS_g0l0_1}) and (\ref{eq:fairZDS_g0l0_2}).
In order to check the validity, we perform numerical simulation for the case $\left( \bar{R}, \bar{S}, \bar{T}, \bar{P} \right)=(3, 0, 5, 1)$ and $\left( \delta_{R}, \delta_{S}, \delta_{T}, \delta_{P} \right)=(-3, 1, 1, 4)$ (satisfying (\ref{eq:necsuf_2-1})).
$c_1$ is given by the equality condition of Eq. (\ref{eq:c1_g0l0}).
The result is given in Figure \ref{fig:linear_ZDS_type2}.
This result is consistent with Theorem \ref{thm:necsuf_fair}.
\begin{figure}
\begin{center}
\includegraphics[clip, width=6.0cm]{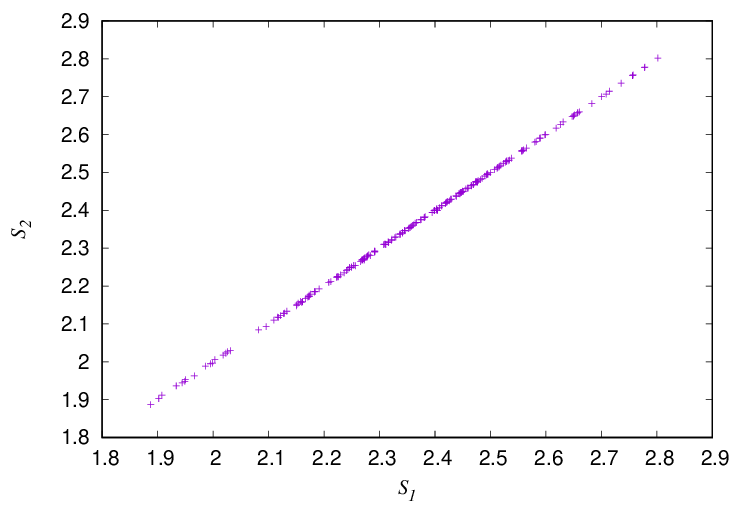}
\end{center}
\caption{A linear relation between $\left\langle s_1 \right\rangle^{*}$ and $\left\langle s_2 \right\rangle^{*}$ when player $1$ uses the fair ZD strategy (\ref{eq:fairZDS_g0l0_1}) and (\ref{eq:fairZDS_g0l0_2}) with $c_1$ satisfying the equality condition of Eq. (\ref{eq:c1_g0l0}) and player $2$ uses $200$ randomly generated memory-one strategies. The payoffs are set to $\left( \bar{R}, \bar{S}, \bar{T}, \bar{P} \right)=(3, 0, 5, 1)$ and $\left( \delta_{R}, \delta_{S}, \delta_{T}, \delta_{P} \right)=(-3, 1, 1, 4)$. Each $\left\langle s_j \right\rangle^{*}$ is calculated by time average over $10^6$ time steps. The fair ZD strategy indeed enforces a linear relation $\left\langle s_1 \right\rangle^{*}=\left\langle s_2 \right\rangle^{*}$.}
\label{fig:linear_ZDS_type2}
\end{figure}

\subsection{The case (\ref{eq:necsuf_2-2})}
For this case, a fair ZD strategy is given by Eqs. (\ref{eq:fairZDS_l0g0_1}) and (\ref{eq:fairZDS_l0g0_2}).
In order to check the validity, we perform numerical simulation for the case $\left( \bar{R}, \bar{S}, \bar{T}, \bar{P} \right)=(3, 0, 5, 1)$ and $\left( \delta_{R}, \delta_{S}, \delta_{T}, \delta_{P} \right)=(4, 1, 1, -3)$ (satisfying (\ref{eq:necsuf_2-2})).
$c_1$ is given by the equality condition of Eq. (\ref{eq:c1_l0g0}).
The result is given in Figure \ref{fig:linear_ZDS_type3}.
This result is consistent with Theorem \ref{thm:necsuf_fair}.
\begin{figure}
\begin{center}
\includegraphics[clip, width=6.0cm]{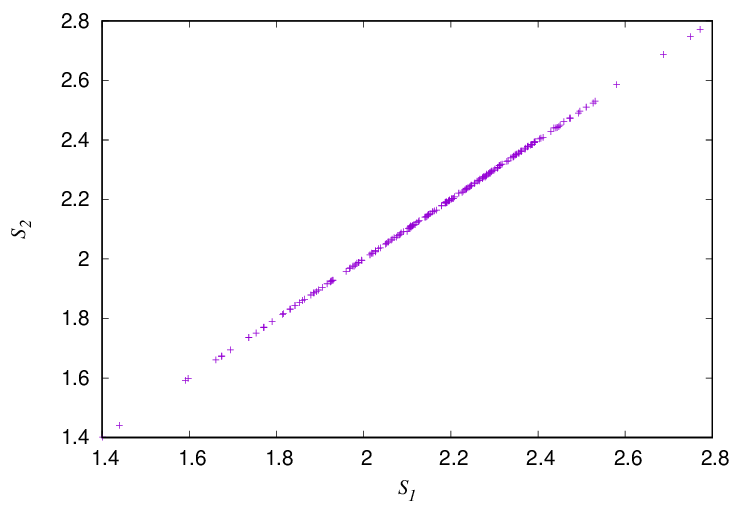}
\end{center}
\caption{A linear relation between $\left\langle s_1 \right\rangle^{*}$ and $\left\langle s_2 \right\rangle^{*}$ when player $1$ uses the fair ZD strategy (\ref{eq:fairZDS_l0g0_1}) and (\ref{eq:fairZDS_l0g0_2}) with $c_1$ satisfying the equality condition of Eq. (\ref{eq:c1_l0g0}) and player $2$ uses $200$ randomly generated memory-one strategies. The payoffs are set to $\left( \bar{R}, \bar{S}, \bar{T}, \bar{P} \right)=(3, 0, 5, 1)$ and $\left( \delta_{R}, \delta_{S}, \delta_{T}, \delta_{P} \right)=(4, 1, 1, -3)$. Each $\left\langle s_j \right\rangle^{*}$ is calculated by time average over $10^6$ time steps. The fair ZD strategy indeed enforces a linear relation $\left\langle s_1 \right\rangle^{*}=\left\langle s_2 \right\rangle^{*}$.}
\label{fig:linear_ZDS_type3}
\end{figure}

\section{Equalizer strategy in a special example}
\label{app:equalizer}
As a special example, we consider the situation \cite{MMHet2025}
\begin{align}
 \bm{s}_1 &= \left(
\begin{array}{c}
b \\
\frac{b-c}{2} \\
b \\
\frac{b-c}{2} \\
0 \\
0 \\
\frac{b}{2} \\
\frac{b}{2}
\end{array}
\right),
\quad
 \bm{s}_2 = \left(
\begin{array}{c}
0 \\
\frac{b}{2} \\
0 \\
\frac{b}{2} \\
b \\
b \\
\frac{b-c}{2} \\
\frac{b-c}{2}
\end{array}
\right)
\end{align}
with $b>0$ and $c>0$, that is, $\left( R^{(1)}, S^{(1)}, T^{(1)}, P^{(1)} \right) = \left( b, (b-c)/2, b, (b-c)/2 \right)$ and $\left( R^{(2)}, S^{(2)}, T^{(2)}, P^{(2)} \right) = \left( 0, 0, b/2, b/2 \right)$.
We look for equalizer strategies of player $1$, which correspond to
\begin{align}
 \bm{B} &= \bm{s}_2 - r\bm{1},
\end{align}
where $r\in \mathbb{R}$ and $\bm{1}$ is the vector of all ones.
When we set
\begin{align}
 T_1(C | \sigma_2,  C,C,\sigma_1) &= \frac{4r-(b-c)}{b+c} \nonumber \\
 T_1(C | \sigma_2,  C,D,\sigma_1) &= \frac{4r-(2b-c)}{b+c} \nonumber \\
 T_1(C | \sigma_2,  D,C,\sigma_1) &= \frac{4r-(b-c)}{b+c} \nonumber \\
 T_1(C | \sigma_2,  D,D,\sigma_1) &= \frac{4r-(2b-c)}{b+c}
\end{align}
with $b/2-c/4 \leq r \leq b/2$ and $T_1(C | \sigma_1,  \bm{a}^\prime,\sigma_2)$ to arbitrary values, we obtain
\begin{align}
 & -(b-r) \bm{\hat{T}}_1\left( C, \sigma_2 \right) - \left( \frac{b-c}{2} - r \right) \bm{\hat{T}}_1\left( D, \sigma_2 \right) \nonumber \\
 &= -(b-r)\left(
\begin{array}{c}
\frac{4r-(b-c)}{b+c} \\
\frac{4r-(2b-c)}{b+c} \\
\frac{4r-(b-c)}{b+c} \\
\frac{4r-(2b-c)}{b+c} \\
-1 \\
-1 \\
0 \\
0
\end{array}
\right) - \left( \frac{b-c}{2} - r \right)\left(
\begin{array}{c}
1-\frac{4r-(b-c)}{b+c} \\
1-\frac{4r-(2b-c)}{b+c} \\
1-\frac{4r-(b-c)}{b+c} \\
1-\frac{4r-(2b-c)}{b+c} \\
0 \\
0 \\
-1 \\
-1
\end{array}
\right) \nonumber \\
 &= \left(
\begin{array}{c}
-\frac{b+c}{2} \frac{4r-(b-c)}{b+c} - \frac{b-c}{2} + r \\
-\frac{b+c}{2} \frac{4r-(2b-c)}{b+c} - \frac{b-c}{2} + r \\
-\frac{b+c}{2} \frac{4r-(b-c)}{b+c} - \frac{b-c}{2} + r \\
-\frac{b+c}{2} \frac{4r-(2b-c)}{b+c} - \frac{b-c}{2} + r \\
b-r \\
b-r \\
\frac{b-c}{2} - r \\
\frac{b-c}{2} - r
\end{array}
\right) \nonumber \\
  &= \bm{s}_2 - r\bm{1}.
\end{align}
Therefore, this memory-one strategy is an equalizer strategy, which unilaterally enforces $\left\langle s_2 \right\rangle^{*} = r$.

\section*{References}

\bibliography{ZDS}

\end{document}